\documentclass{article}

\usepackage{amsmath,amssymb,natbib,amsthm,amsbsy}
\usepackage{times}
\RequirePackage[colorlinks,citecolor=blue,urlcolor=blue]{hyperref}
\usepackage{graphicx}
\usepackage{epsfig}
\usepackage{subfigure}
\usepackage{float}
\usepackage{multirow}
\usepackage{enumerate}

\usepackage{color}
\usepackage[dvipsnames]{xcolor}

\usepackage{tikz}
\usetikzlibrary{positioning}
\usetikzlibrary{arrows,automata}
\tikzset{
  supset/.style={
    draw=none,
    edge node={node [sloped,  auto=false]{$\supset$}}},
  Subset/.style={
    draw=none,
    every to/.append style={
      edge node={node [sloped, allow upside down, auto=false]{$\subset$}}}
  }
}

\newtheorem{corollary}{Corollary}[section]
\newtheorem{proposition}{Proposition}[section]
\newtheorem{definition}{Definition}[section]
\newtheorem{algorithm}{Algorithm}

\def \S {{\mathcal{S}}}
\def \H {{\mathcal{H}}}

\newcommand{\R}{\mathbb{R}}
\newcommand{\D}{\mathcal{D}}

\newcommand {\argmin} {\mathop{\rm argmin}}
\newcommand{\tr}{\mathop{\rm tr}}

\usepackage[ulem=normalem]{changes}

\begin{document}

\title{Smoothing splines on Riemannian manifolds, with applications to 3D shape space}

\author{Kwang-Rae Kim$^*$, Ian L.~Dryden$^+$, Huiling Le$^+$ and Katie E.~Severn$^+$\\
$^*$SAS Korea and  $^+$University of Nottingham}

\date{}

\maketitle

\begin{abstract}
There has been increasing interest in statistical analysis of data lying in manifolds. This paper generalizes a smoothing spline fitting method to Riemannian manifold data based on the technique of unrolling, unwrapping and wrapping originally proposed in \cite{JK} for spherical data. In particular we develop such a fitting procedure for shapes of configurations in general $m$-dimensional Euclidean space, extending our previous work for  two dimensional shapes. We show that parallel transport along a geodesic on Kendall shape space is linked to the solution of a homogeneous first-order differential equation, some of whose coefficients are implicitly defined functions. This finding enables us to approximate the procedure of unrolling and unwrapping by simultaneously solving such equations numerically, and so to find numerical solutions for smoothing splines fitted to higher dimensional shape data. This fitting method is applied to the analysis of some dynamic 3D peptide data. 
\end{abstract}

\vskip 0.5cm 

\begin{small}
\end{small}

\section{Introduction}

Analysis of temporal shape data has become increasingly important for applications in many fields,
and a common definition is that the shape of an object is what is left after removing the effects of rotation, translation and re-scaling 
\citep{Kendall84}. For an introduction to the statistical analysis of shape see \citet{Drydmard16,KBCL,Patrelli16} and \citet{Srivklas16}.
An important focus is on the shapes of landmark data, where each object consists of $k > m$ points in $\mathbb{R}^m$. After removing 
translation, scale and rotation the data lie in Kendall's shape space, denoted as $\Sigma^k_m$ \citep{Kendall84}.

Suppose that we are given a time series of landmark shape data which are measurements of a moving object.
One may wish to find the overall trend or be interested in the behaviour of the object at unobserved times,
including an explanation of how the shapes change between successive times.
For example, in the field of molecular biology, the study of dynamic proteins is of interest. However, currently there is very limited methodology available for fitting models for $m=3$ dimensional shape data which exhibit a large amount of 
variability. 

The main difficulty in applying classical statistical approaches directly to landmark shape data lies in the fact that shape spaces are 
curved manifolds and have singularities if $m=3$, so that standard multivariate linear methods are not appropriate. 
In particular, if we wish to fit a smoothing spline through points on a curved manifold, as pointed out by \citet{JK},  one cannot simply use a spline defined by 
linear combinations of manifold valued basis functions. Such linear operations are not available in general on a manifold. 
 \citet{JK}'s novel approach to the problem for spline fitting on a sphere was to unroll a base path onto a plane; unwrap the data onto the plane with respect to the base path; fit a smoothing spline to the unwrapped data in this plane; and finally wrap the fitted spline values back to the sphere. Their method was applied to fitting a smoothing spline to an apparent polar wander path on earth 
from positions of the paleomagnetic north pole taken over time. 

A generalisation of  \citet{JK}'s method was considered by \cite{KDL} for shape data in $m=2$ dimensions and by \cite{Pauley11} for data on general Lie groups 
(called `JK-cubics'). \cite{Pauley11}  also considered `Riemannian cubics' which involve solving the Euler-Lagrange equation.  All these
situations benefit from the manifolds being homogeneous (where local geometrical properties are identical over the manifold). 
However, so far it has not been possible to apply the technique to non-homogenous manifolds. 
In this paper we generalize  \citet{JK}'s method to general Riemannian manifolds, and in particular provide methods and applications for 
the important practical case of temporal sequences of 3D shape data.

If the shape changes of interest are small then an alternative way to make progress is to carry out classical techniques on the projection of the data onto the Procrustes tangent space \citep{KM} at an estimate of the mean of the data, as carried out by \cite{KMMA} and \cite{MKMFAL}. However, if the variability is large then the tangent space projection can be a poor approximation to the manifold, leading to inappropriate inference. If the data approximately follow 
a geodesic path then geodesic based fitting approaches would be appropriate, such as given by \cite{LK}, \cite{FLPJ} and \citet{Fletcher13}.
A more general method of principal geodesic analysis based on intrinsic methods of fitting geodesics to data on Riemannian manifolds was developed in \cite{HZ} and \cite{Hucketal10}. These methods are more widely applicable than the 
principal geodesic analysis proposed in \cite{FLPJ} and \citet{Fletcher13}, 
and provide an alternative definition of a mean. However, these geodesic approaches will be limited when the data follow 
more complicated paths. 
Other relevant work includes regression on Riemannian symmetric spaces \citep{Cornetal17}, intrinsic polynomial regression \citep{Hinkleetal14},  extrinsic local regression \citep{Linetal17}, global and local Fr\'echet regression 
\citep{Petemull19}, intrinsic and varying coefficient regression models for diffusion tensors 
\citep{Yuanetal13,Zhuetal09}. \cite{SDKLS} discussed smoothing splines using a Palais metric-based gradient method 
and,  as for \citet{Pauley11}'s Riemannian cubics, this approach
requires explicit knowledge of the complete
geometrical structure of the manifold, which is often difficult to calculate. None of these methods can be used for fitting general curves to 3D shape data, which is the main motivation for this paper.

In Section \ref{Sec2} we first provide an introduction to the basic geometrical ideas and then
generalize \citet{JK}'s method to Riemannian manifolds. In Section \ref{Sec3} we investigate how to use the result in \cite{Le03} to analyse three or more dimensional shape data from a more practical point of view. An important contribution in this paper is that we give an explicit method 
for computing parallel transport for landmark shapes in at least three dimensions. \citet{Le03} discussed 
parallel transport in quotient spaces in general, and in particular 
her Theorem 1 gives three mathematical conditions which need to be 
satisfied for parallel transport in shape space. For landmarks shapes in 
$m=2$ dimensions explicit expressions are available which satisfy the 
conditions, and were utilized by \citet{KDL} for fitting smoothing 
splines. {In contrast,} \citet{Le03}'s three conditions involve the unknown 
parallel transport vector $V(s)$ and some unknown skew-symmetric matrices 
$A(s)$ but no quantitative method for how to find such matrices and construct parallel transport in the general case. Although there was some discussion in 
special cases, no general method was given to construct a solution. 
In our paper we give Proposition \ref{Prop2}, with a proof, which enables 
us to find the $A(s)$ and compute the parallel transport $V(s)$ by numerically 
solving a system of homogeneous first-order ordinary differential equations (ODEs), some of whose coefficients are implicitly defined functions.   
So, for the first time, we are now able to compute 
the parallel transport for 3D landmark shapes. 

In addition we implement the generalisation
of \citet{JK}'s smoothing splines using unrolling, unwrapping and wrapping 
to 3D shape space, and an overview of an algorithm is given in Section \ref{algorithm}. By solving the ODEs numerically we can approximate the procedure of unrolling and unwrapping, and hence obtain numerical solutions for smoothing splines fitted to three dimensional shape data.
This is the first time that the unrolling and unwrapping procedure has been implemented for a non-homogenous quotient space as far as we are aware. 
Hence, our methodology has provided a 
{significant increase} in generality
to the very important practical case of analyzing 3D shape data. In Section \ref{Sec4} we apply the methodology to the analysis of 3D molecular 
dynamics data. Molecular dynamics data has been used in a wide variety of important applications, for example in the study of
protein folding and enzyme catalysis \citep{Karplus05}. A key component is the study of the 3D shape of a molecule as it changes 
over time. Finally in Section \ref{Sec5} we conclude the paper with a brief discussion.

\section{Smoothing splines on manifolds}\label{Sec2}
\subsection{Unrolling, unwrapping and wrapping}
Recall that three of the main ingredients of the spline fitting technique used in \cite{JK} and \cite{KDL} are respectively $(i)$ unrolling a curve, or path, in the relevant space to the tangent space at its starting point; $(ii)$ unwrapping points in the space at known times, with respect to a 
path, onto the tangent space at the starting point of the path; and $(iii)$ wrapping points onto a manifold, which is the reverse of unwrapping. 
{The path, with respect to which the unwrapping is carried out, is usually called the `base path'. To} generalize the spline fitting of \cite{JK} to general manifolds, we reformulate these three procedures in a more general manifold setting.

We briefly provide an introduction to some relevant aspects of differential geometry, and for further details there are many texts, including \citet{Bar10},  \citet{Boothby86} and \citet[Section 3.1]{Drydmard16}.
Informally a Riemannian manifold $M$ is a manifold that has a positive-definite inner product (Riemannian metric tensor $g$) defined on the tangent space at each $x \in M$, which varies smoothly with $x$.  The Riemannian metric tensor allows one to define geometrical properties of the manifold, including measuring distance along a shortest path (minimal geodesic) between two points on the manifold. Write $\tau_{x}(M)$ 
for the tangent space of $M$ which touches the manifold at $x$. The exponential map is a map from the 
tangent space to the manifold $M$ and the definition is given below. Its inverse is called the inverse exponential map. 
The left plot of Figure \ref{tanfig}  shows the case of a sphere for $M$, geodesic $\gamma$, tangent space $\tau_x(M)$, distance $d$, and the inverse exponential map projection from the manifold $x' \in M$ to  a point $q$ in the tangent space. 
The Euclidean 
distance $d$ between the pole $x$ and $q$ in the tangent space is the same length as the induced Riemannian distance 
$d(x,x')$ between $x$ and $x'$ along the geodesic $\gamma$, and this is a particular property of the exponential map for a Riemannian manifold.

\begin{figure}[htbp]
\centering
\begin{minipage}[b][5cm][s]{.45\textwidth}
\begin{tikzpicture}
[
  point/.style = {draw, circle, fill=black, inner sep=0.7pt},
]
\def\rad{2cm}
\coordinate (O) at (0,0); 
\coordinate (N) at (0,\rad);
\coordinate (R) at (0.923*\rad-10,0.382*\rad); 
\coordinate (TR) at (0.923*\rad,1.4*\rad-20);
\coordinate (RR) at (-0.923*\rad+15,0.282*\rad); 
\filldraw[ball color=white] (O) circle [radius=\rad];
\draw[dashed] 
  (\rad,0) arc [start angle=0,end angle=180,x radius=\rad,y radius=5mm];
\draw
  (\rad,0) arc [start angle=0,end angle=-180,x radius=\rad,y radius=5mm];
  \draw[black]
  (R) arc [start angle=22.5,end angle=80,x radius=\rad,y radius=20mm];
  \node[black] at (1,1.3) {$\gamma$};  
\begin{scope}[xslant=0.5,yshift=\rad,xshift=-2]
\filldraw[fill=gray!50,opacity=0.3]
  (-4,1) -- (3,1) -- (3,-1) -- (-4,-1) -- cycle;
\node at (-3,0.6) {$\tau_{{x}}(M)$};  
\end{scope}
\draw[dashed]
  (N) node[above] {$x$} -- (O) node[below] { };
\draw[black]
  (N) node[above] {} -- node[above] {$d$}(TR) node[right] {};
\draw (R)node[left]{$x'$};
\draw (TR)node[right]{$q$};
\path[->]
(R) edge [bend right] node [right]{$\exp_x^{-1}$} (TR)[dashed];
\node[point] at (N) {};
\node[point] at (R) {};
\node[point] at (TR) {};
\draw (RR)node[left]{$M$};
\end{tikzpicture}
\end{minipage}\qquad
\begin{minipage}[b][5cm][s]{.45\textwidth}\centering

\begin{tikzpicture}
[
  point/.style = {draw, circle, fill=black, inner sep=0.7pt},
]
\def\rad{2cm}
\coordinate (O) at (0,0); 
\coordinate (N) at (0,\rad); 
\coordinate (R) at (\rad-32,0.4); 
\coordinate (TR) at (0.923*\rad,1.4*\rad-20);
\coordinate (RR) at (-0.923*\rad+15,0.282*\rad);
\coordinate (O) at (0,0); 
\coordinate (B) at (0,\rad); 
\coordinate (C) at ({0.6*\rad*cos(30)},{0.6*\rad*sin(30)}); 
\coordinate (D) at (0,-\rad);
\coordinate (E) at ({-0.6*\rad*cos(30)+5},{0.6*\rad*sin(30)}); 
\filldraw[ball color=white] (O) circle [radius=\rad];
\draw[dashed] 
  (\rad,0) arc [start angle=0,end angle=180,x radius=\rad,y radius=5mm];
  %
%
\draw[dashed]
  (\rad,0) arc [start angle=0,end angle=-180,x radius=\rad,y radius=5mm];       
  %
    \draw[black]
  (B) arc [start angle={87},end angle=17,x radius={\rad*0.6},y radius=20mm] 
  \foreach \p in {0,10,...,100} {
      node[sloped,pos=\p*0.01]
      (N2 \p){}
    };
    \foreach \p in {0,30,60,90} {
      \draw[-latex, black, thick, shorten >= -0.5cm] (N2 \p.center) -- (N2 \p.east);
    };
\draw
  (C) arc [start angle=0,end angle=-180,x radius={0.6*\rad*cos(30)-2.5},y radius=1mm]
  \foreach \p in {0,5,...,100} {
      node[sloped,pos=\p*0.01]
      (N \p){}
    };
    \foreach \p in {30,50,70} {
  \draw[-latex, black, thick, shorten >= -0.5cm] (N \p.center) -- (N \p.south);
   };
%
  \draw[black]
  (E) arc [start angle=163,end angle=230,x radius={\rad*0.6-3.5},y radius=20mm]
    \foreach \p in {0,5,...,100} {
      node[sloped,pos=\p*0.01]
      (N3 \p){}
    };
        \foreach \p in {20} {
      \draw[-latex, black, thick, shorten >= -0.5cm] (N3 \p.center) -- (N3 \p.west);
    };
            \foreach \p in {60, 90} {
      \draw[-latex, black, thick, shorten >= -0.5cm] (N3 \p.center) -- (N3 \p.east);
    };    
    %
    %
  %
\begin{scope}[xslant=0.5,yshift=\rad,xshift=-2]
\filldraw[fill=gray!50,opacity=0]
  (-4,1) -- (3,1) -- (3,-1) -- (-4,-1) -- cycle; 
\end{scope}
\node[point] at (B) {};
\node[point] at (C) {};
\node[point] at (E) {};
\node[point] at (-0.62,-1.35) {};
\draw (B)node[above]{$\scriptstyle{A}$};
\draw (C)node[right]{$\scriptstyle{B}$};
\draw (E)node[left]{$\scriptstyle{C}$};
\draw (-0.7,-1.3)node[right]{$\scriptstyle{D}$};
\draw (RR)node[left]{$M$};
\end{tikzpicture}
\end{minipage}
\caption{\footnotesize{\textit{(left) A diagrammatic view of $M$ in the case of a sphere, with tangent space 
$\tau_x(M)$ which touches the manifold at $x$; the minimal geodesic $\gamma$ from $x$ to $x'$ which has the same 
length as the Euclidean distance $d$ in the tangent space from $x$ to $q$; and the inverse exponential map projection from the manifold $x' \in M$ to $q$ in the tangent space
\citep[adapted from][]{SDP}. (right) An illustration of parallel transport. The vector in the tangent space at $A$ is parallel transported 
along the geodesic to $B$, then along the geodesic to $C$ and then finally along the geodesic to $D$.
 }}}\label{tanfig}
\end{figure}

In more detail let $M$ be a complete and connected Riemannian manifold with induced Riemannian distance $d$ and denote by $\tau_x(M)$ the tangent space to $M$ at $x\in M$. 
The exponential map on $M$ at $x$ can be expressed by
\[\exp_x:\tau_x(M)\mapsto M;\,\,\,\exp_x(v)=\gamma(\|v\|),\]
where $\gamma$ is the 
geodesic such that $\gamma(0)=x$ and $\dot\gamma(0)=v/\|v\|$. {Since $\|\dot\gamma(0)\|=1$, $\gamma$ is usually referred to as a unit-speed geodesic.} Thus, when there is a unique shortest unit-speed geodesic $\gamma$ from $x$ to $x'$, the inverse of the exponential map can be expressed as 
\[\exp_x^{-1}(x')= d(x,x')\,\dot\gamma(0) = q \in \tau_x(M) ,\]
where $d(x,x')$ is the induced Riemannian distance between $x$ and $x'$.

Intuitive descriptions of the unrolling, unwrapping and wrapping procedures can be found in \cite{JK} for spherical data and \cite{KDL} for planar shape data. The basic idea can be understood by considering the manifold $M$ as a sphere with a 
base path $\gamma(t) \in [t_0, t_n]$ marked in wet ink on the sphere, and a base tangent plane which touches the sphere at the start point of the base path 
$\gamma(t_0)$.   The unrolling of 
the path $\gamma(t)^\dagger$ is the trace that the wet ink leaves on the tangent plane after rolling the tangent plane along the base path on the sphere without slipping or twisting so that at time $t$ the tangent plane touches 
the sphere at $\gamma(t)$. The unwrapping of a point on the sphere with respect to $\gamma$ is then a projection into the base tangent space 
that preserves direction and distance. The wrapping from the base tangent space to the manifold is the reverse of the unwrapping procedure. 
The unrolling, unwrapping and wrapping ideas are also appropriate when $M$ is a Riemannian manifold, although their construction is in 
general more complicated.

\begin{figure}[htbp]
\begin{center}
    {\includegraphics[trim={5cm 4cm 0cm 4cm}, clip,scale = 0.5]{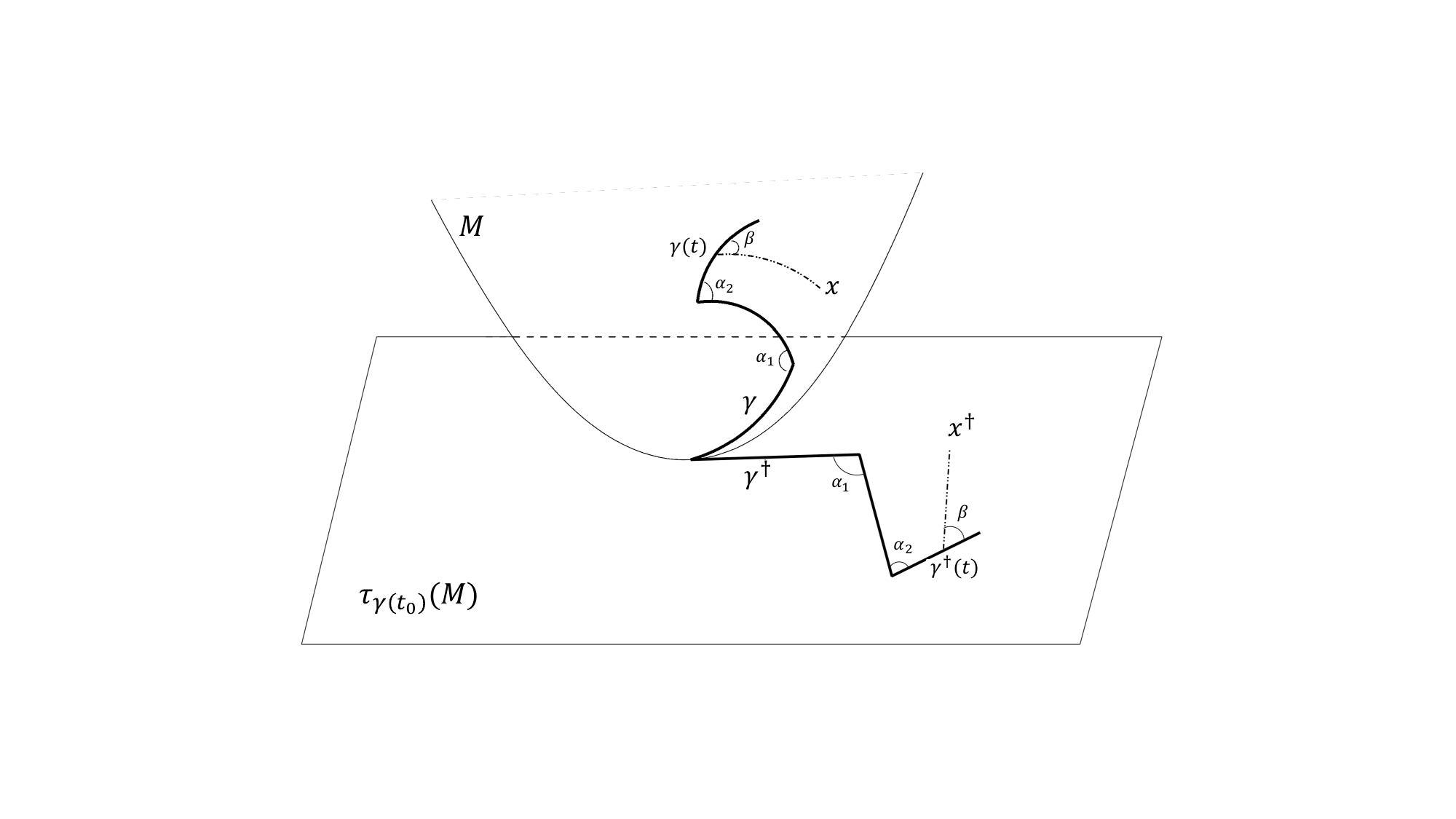}}   
   \end{center}
   \caption{A diagrammatic view of unrolling, unwrapping and wrapping \citep[adapted from][]{KDL}.} 
   \label{unroll}
\end{figure}

In Figure \ref{unroll} we provide a diagrammatic view of the unrolling, unwrapping and wrapping 
in the special case where $\gamma(t)$ is a continuous piecewise geodesic curve in a Riemannian manifold $M$.  
The piecewise geodesics in 
{$M$}, the piecewise linear paths
in tangent space and the angles are displayed.
In particular, the unrolling of a geodesic 
in $M$ on to the tangent space at its
initial 
{point} is {a} 
straight-line segment, starting from the origin{. This linear segment}  
is determined by the tangent vector to the geodesic at its initial point and {has} 
the same length as the
geodesic. 
The unrolling of a continuous piecewise 
geodesic 
on to
the tangent space at its initial 
{point} is a piecewise linear {path}, 
 starting from the origin{. The} 
first segment {of this piecewise linear path} has the same length as the first geodesic segment and its direction
is determined by the tangent vector to the first piece of the geodesic at its initial point{. The} 
lengths of the remaining segments {of this piecewise linear path} are the same as the corresponding parts of the
piecewise geodesic 
and the angles through which they turn are the same as those
of the corresponding parts of the piecewise geodesic. 
Moreover, the unwrapping
at time $t$, with respect to such a 
{piecewise geodesic} $\gamma$ of a point 
{$x$} is the point 
${x}^\dagger \in \tau_{\gamma(t_0)}(M)$ 
{determined as follows:} the length of 
${x}^\dagger - \gamma(t)^\dagger$ is the same as that of the geodesic from
$\gamma(t)$ to 
{$x$} and the angle that 
${x}^\dagger - \gamma(t)^\dagger$ makes with $\gamma^\dagger$
is the same as the angle that the geodesic from $\gamma(t)$ to 
{$x$} makes with 
$\gamma$.   Finally the wrapping at time $t$ is the reverse of the unwrapping, and maps $x^\dagger$ on  $\tau_{\gamma(t_0)}(M)$ to $x \in M$.

Formally,  unrolling, unwrapping and wrapping are closely linked to the concept of parallel transport along a curve. The latter is a method of transporting tangent vectors along smooth curves in a manifold and, in some sense, provides a method of isometrically moving the local geometry of a manifold along  curves. 
The Riemannian metric tensor $g$ defines a unique affine connection called the {Levi-Civita connection} $\nabla$, or covariant derivative, 
which enables us to say how vectors in tangent spaces change as we move along a curve from one point to another without slipping or twisting.   
Let $\gamma(t)$ be a curve on $M$ 
and we want to move from $\gamma(t_0)$ to $\gamma(s)$ and see how the vector $v(\gamma(t_0)) \in \tau_{\gamma(t_0)}(M)$ 
is transformed to in the new tangent space $\tau_{\gamma(s)}(M)$. The {parallel transport}  
of $v$ along $\gamma$ is a vector field $v(s)$ which satisfies a system of differential equations in (\ref{pteq}) which 
in general is solved numerically, although for
some manifolds the solution is analytic.     
The right plot of Figure \ref{tanfig}  illustrates the parallel transport of a vector along a continuous piecewise geodesic path. 
The vector in the tangent space at $A$ is parallel transported along the geodesic to $B$, then along the geodesic to $C$ and then finally along the geodesic to $D$.

More precisely, suppose that $\gamma(t)$, $0\leqslant t\leqslant t^*$ is a smooth curve in $M$ and that $v_0\in\tau_{\gamma(0)}(M)$, where $t^* > 0$ is a 
fixed constant. Then, the parallel transport of $v_0$ along $\gamma$ is the extension of $v_0$ to a vector field $v$ on $\gamma$ such that
\begin{equation}
\nabla_{\dot\gamma(t)} v(\gamma(t)) = 0,\qquad v(\gamma(0)) = v_0 .  \label{pteq}
\end{equation}
This provides linear isomorphisms between the tangent spaces at points along the curve $\gamma$:
\[P_{\gamma(s)}^{\gamma(t)} : \tau_{\gamma(t)}(M) \rightarrow \tau_{\gamma(s)}(M).\]
This isomorphism is known as the  {(Levi-Civita)} parallel transport map associated with the curve. Note that $\left(P_{\gamma(s)}^{\gamma(t)}\right)^{-1}=P_{\gamma(t)}^{\gamma(s)}$.

\begin{definition}
In terms of the parallel transport $P_{\gamma(s)}^{\gamma(t)}$  for a smooth curve $\gamma$ in $M$:
\begin{enumerate}
\item[$\bullet$] the {\bf unrolling} of $\gamma$ onto $\tau_{\gamma(0)}(M)$ is the curve $\gamma^\dagger$ on $\tau_{\gamma(0)}(M)$ such that $\gamma^\dagger(0)=0$ and 
\[\frac{d\gamma^\dagger(t)}{dt}=P_{\gamma(0)}^{\gamma(t)}\left(\dot\gamma(t)\right),\qquad 0\leqslant t\leqslant t^*;\]
\item[$\bullet$] the {\bf unwrapping} at $t\in[0,t^*]$, with respect to $\gamma$, of a point $x\in M$ into $\tau_{\gamma(0)}(M)$  {is the tangent vector at $\gamma(0)$ obtained as follows: first map $x$ to $\tau_{\gamma(t)}(M)$ using $\exp^{-1}_{\gamma(t)}$ to give $\exp^{-1}_{\gamma(t)}(x)$; then parallel translate it along $\gamma(t)$ back to $\tau_{\gamma(0)}$ to give the tangent vector $P^{\gamma(t)}_{\gamma(0)}(\exp^{-1}_{\gamma(t)}(x))$ at $\gamma(0)$; finally, the unwrapping at $t$ of $x\in M$ into $\tau_{\gamma(0)}(M)$ is the (Euclidean parallel) translation, within $\tau_{\gamma(0)}(M)$, of the tangent vector $P^{\gamma(t)}_{\gamma(0)}\left(\exp^{-1}_{\gamma(t)}(x)\right)$ to $\gamma^\dagger(t)$. In other words, the unwrapping at $t$ of $x\in M$ into $\tau_{\gamma(0)}(M)$} 
is the sum of the two tangent vectors in $\tau_{\gamma(0)}(M)$: $\gamma^\dagger(t)$ and the parallel transport of the tangent vector $\exp_{\gamma(t)}^{-1}(x)\in\tau_{\gamma(t)}(M)$ along $\gamma$ to $\tau_{\gamma(0)}(M)$, i.e. the unwrapping at $t$ of $x$, with respect to $\gamma$, is 
\[\gamma^\dagger(t)+P_{\gamma(0)}^{\gamma(t)}\left(\exp_{\gamma(t)}^{-1}(x)\right);\]
\item[$\bullet$] the {\bf wrapping} at $t$, with respect to $\gamma$, of a tangent vector $v\in\tau_{\gamma(0)}(M)$ back into $M$ is the reverse of the unwrapping procedure.
\end{enumerate}
\end{definition}

In particular, if $\gamma$ is a geodesic, then $\gamma^\dagger$ is the linear segment of the same length as $\gamma$ and in the same direction as the initial tangent vector of $\gamma$, that is, it is the linear segment from the origin of $\tau_{\gamma(0)}(M)$ to $\exp_{\gamma(0)}^{-1}(\gamma(t^*))$.

\subsection{Smoothing spline fitting on manifolds}
The goal in our work is to fit a smooth path to a dataset of points on a manifold that are recorded at times $t_0,t_1,\ldots,t_n$. 
The procedure involves a trade-off between fitting the observed data well and the curve being smooth. The fitted smooth path is called a smoothing spline and is useful for prediction, for example providing a smooth estimate of the trend in the presence of noisy data and for interpolating on the manifold at times between observed data points.

At first let us recall information about smoothing splines in Euclidean space, for example see {\citet{Greesilv94}. 
For $v_0, v_1, \ldots, v_n \in \R^d$ observed at times, $t_0, t_1, \ldots, t_n$ respectively, the cubic smoothing spline 
{estimation} for this dataset in $\R^d$ 
{seeks} the function{, or path,} $\widehat{f_\lambda^\dagger}:T\rightarrow\mathbb{R}^d$ that minimizes
\[\sum_{i = 0}^n \left\| v_i - f_\lambda^\dagger(t_i) \right\|^2+ \lambda \int_T \left\| (f_\lambda^\dagger)''(t) \right\|^2 dt\]
among all $\mathcal{C}^2$-functions, where $T = [t_0, t_n]$ and $\lambda > 0$ denotes a smoothing parameter {\citep{Greesilv94}.
The integrated squared norm of second derivative of the function is a measure of the roughness of the curve. 
Since the objective function is additive over each  component of $f_\lambda^\dagger$, this} 
minimization problem can be solved on each component ${f^\dagger}^{(j)}$ of $f^\dagger$ separately to reduce the $d$-dimensional minimization problem to $d$ one dimensional ones, so that
\begin{equation}\label{cubicspline}
 \widehat{{f_\lambda^\dagger}^{(j)}}= \argmin_{{f_\lambda^\dagger}^{(j)}}\sum_{i = 0}^n \left\{ v_i^{(j)} - {f_\lambda^\dagger}^{(j)} (t_i) \right\}^2+ \lambda \int_T \left | \left({f_\lambda^\dagger}^{(j)}\right)''(t) \right |^2 dt,
\end{equation} 
for $j = 1, \ldots,d$.
The solution to such a minimization problem is a cubic spline with knots at the unique values of $v_i$, while the smoothing parameter $\lambda$ controls the trade-off of the model complexity. When $\lambda$ tends to zero, $\widehat{f_\lambda^\dagger}$ becomes an interpolating cubic spline. On the other hand, $\widehat{f_\lambda^\dagger}$ becomes a classical simple linear regression line when $\lambda$ tends to infinity. 

Penalties with other functional forms could also be used and 
the knots do not need to coincide with the data points. Using $p \ge 1$ order derivatives in the penalty leads to a spline of order $2p-1$  \citep{Schoenberg64}. The integrated squared second derivative penalty function ($p=2$) used above leads to the cubic smoothing spline which often works well in practice. We also consider the $p=1$ case
\begin{equation}
\widehat{{f_{L,\lambda}^\dagger}^{(j)}}= \argmin_{{f_\lambda^\dagger}^{(j)}}\sum_{i = 0}^n \left\{ v_i^{(j)} - {f_\lambda^\dagger}^{(j)} (t_i) \right\}^2+ \lambda \int_T \left | \left({f_\lambda^\dagger}^{(j)}\right)'(t) \right |^2 dt ,  \label{linearspline}
\end{equation}
for $j=1,\ldots,d$. 
The solution is a linear spline for each component, which consists of a continuous piecewise linear function with possible jumps in the
slope at each knot.

{Turning to smoothing splines on manifolds, in} 
a similar fashion to that in \cite{JK} and \cite{KDL}, we use the concept of unrolling and unwrapping to define $M$-valued smooth splines as follows.

\begin{definition}
For a given dataset $\mathcal{D}=\{x_j\,:\, 0\leqslant j\leqslant n\}$ in $M$, where $x_j$ is observed at time $t_j$, and smoothing parameter $\lambda$, we define {\bf the $M$-valued smoothing spline} fitted to $\mathcal D$ with parameter $\lambda$ to be the $\mathcal{C}^2$-function
\[f(\cdot,\lambda) : [t_0, t_n] \rightarrow M\]
such that its unrolling $f^\dagger$ onto $\tau^{\phantom{A}}_{f(t_0,\lambda)}(M)$
is the cubic smoothing spline fitted to the data $\mathcal{D}^\dagger$ obtained by unwrapping $\mathcal D$ at times $t_j$, with respect to $f(\cdot,\lambda)$, into $\tau^{\phantom{A}}_{f(t_0,\lambda)}(M)$.
\label{def1}
\end{definition}

{Similar to spherical case in \cite{JK}, $f(\cdot,\lambda)$ is the smoothing spline with parameter $\lambda$ for $\mathcal D$ if, except at the data times, its unrolling $f^\dagger$ satisfies the equation that $d^4f^\dagger(t)/dt^4=0$, i.e. $f^\dagger$ is the Euclidean smoothing cubic spline with parameter $\lambda$ for $\mathcal D^\dagger$}.

{Note also that, when $\lambda$ tends to infinity, $f(\cdot,\lambda)$ becomes the geodesic segment obtained using the principal geodesic analysis developed in \cite{HZ} and \cite{Hucketal10}.}

To find the smoothing spline fitted to a given dataset $\mathcal D$, the iterative scheme for approximation given in \cite{JK} can be applied to general manifolds: take the piecewise geodesic passing through the data points as the initial curve $f_0$; at each step $\ell\geqslant0$, unwrap $\mathcal D$ with respect to $f_\ell$ to get $\mathcal{D}^\dagger_\ell$ in $\tau^{\phantom{A}}_{f_\ell(t_0)}(M)$; fit a Euclidean smoothing spline $f^\dagger_{\ell+1}$ to $\mathcal{D}^\dagger_\ell$ in $\tau^{\phantom{A}}_{f_\ell(t_0)}(M)$; then   {the curve} $f_{\ell+1}$ is {defined to be} the wrapped path, with respect to $f_\ell$, of $f^\dagger_{\ell+1}$ in $M$. The iterative procedure stops when $f_\ell$ and $f_{\ell+1}$ are sufficiently close.

In such an iterative scheme, smooth curves are approximated by piecewise geodesics. Thus, for the implementation in practice, it is sufficient to restrict ourselves to the procedures of unrolling a piecewise geodesic and of unwrapping and wrapping with respect to a piecewise geodesic. 

\begin{proposition}
For a given curve $\gamma(t)$, $t_0\leqslant t\leqslant t_n$, on $M$, which is a geodesic between $t_i$ and $t_{i+1}$, where $t_0<t_1<\cdots<t_n$, the above unrolling, unwrapping and wrapping procedures along $\gamma$ can be formulated as follows.
\begin{enumerate}
\item[$\bullet$] The {unrolling} of $\gamma$ onto $\tau_{\gamma(t_0)}(M)$ is the piecewise linear segment in the tangent space $\tau_{\gamma(t_0)}(M)$ joining the following successive points:
\[\gamma(t_i)^\dagger
\!=
\!\left\{
\begin{array}{lll}
\displaystyle
\!\!0 & \!\!\!\!\!\!\!\!\!\! i = 0
\\ \\
\!\!\displaystyle
\gamma(t_{i - 1})^\dagger
+P_{\gamma(t_0)}^{\gamma(t_1)} \circ P_{\gamma(t_1)}^{\gamma(t_2)} \circ \ldots \circ
P_{\gamma(t_{i - 2})}^{\gamma(t_{i - 1})}\left(\exp_{\gamma(t_{i-1})}^{-1}(\gamma(t_i))\right) &\\
&\!\!\!\!\!\!\!\!\!\!\!\!\!\!\!\!\!\!\!\!\!1 \leqslant i \leqslant n.
\end{array}
\right.\]
Note that
\begin{eqnarray*}
\gamma^\dagger(t)
&=&
\gamma(t_j)^\dagger
+ \frac{t - t_j}{t_{j + 1} - t_j} \left\{ \gamma(t_{j + 1})^\dagger - \gamma(t_j)^\dagger \right\}\qquad\text{for } t_i<t\leqslant t_{j+1}.
\end{eqnarray*}
\item[$\bullet$] The {unwrapping} at time $t$ of $x\in M$, with respect to $\gamma$, into $\tau_{\gamma(t_0)}(M)$ is the tangent vector $x^\dagger_t\in\tau_{\gamma(t_0)}(M)$ given by 
\begin{eqnarray*}
x^\dagger_t
\!\!\!&=&
\!\!\!\!\!\left\{
\begin{array}{ll}
\displaystyle
\exp_{\gamma(t_0)}^{-1} (x)
&\!\!\!\!\!\!\!\!\!\!\!\!\text{if }t = t_0
\\ \\
\displaystyle
\gamma(t)^\dagger
+P_{\gamma(t_0)}^{\gamma(t_1)} \circ P_{\gamma(t_1)}^{\gamma(t_2)} \circ \cdots \circ P_{\gamma(t_{j - 1})}^{\gamma(t_j)} \circ
P_{\gamma(t_j)}^{\gamma(t)} \left( \exp_{\gamma(t)}^{-1} (x) \right)
&\\
&\!\!\!\!\!\!\!\!\!\!\!\!\!\!\!\!\!\!\!\!\!\!\!\!\!\!\!\!\!\!\!\!\text{if }t_j < t\leqslant t_{j+1},
\end{array}
\right.
\end{eqnarray*}
as long as 
{there is a unique geodesic between $x$ and $\gamma(t)$}.
\item[$\bullet$]
The {wrapping} at time $t \in [t_0, t_n]$, along $\gamma$, of a tangent vector $v\in \tau_{\gamma(t_0)}(M)$ back into $M$ is the point $x_t\in M$ given by
\[x_t=\exp_{\gamma(t)}(v_t),\]
where $v_t$ is the tangent vector in $\tau_{\gamma(t)}(M)$ specified by
\[
v_t
=
\left\{
\begin{array}{ll}
\displaystyle
v&\text{if }  t = t_0
\\
\displaystyle
P_{\gamma(t)}^{\gamma(t_j)} \circ P_{\gamma(t_j)}^{\gamma(t_{j - 1})} \circ \cdots \circ P_{\gamma(t_1)}^{\gamma(t_0)}\left(v - \gamma(t)^\dagger \right)
&\text{if } t_j < t\leqslant t_{j+1}.
\end{array}
\right.\]
\end{enumerate}
\end{proposition}

\vskip 0.5cm
It is clear that, to be able to carry out the unrolling, unwrapping and wrapping procedures in practice, the crucial steps are to find the expressions for the exponential map, as well as its inverse, and for the parallel transport along a geodesic. 

\vskip 0.5cm

{\bf Example:} The expressions for unrolling, unwrapping and wrapping are 
known in the case of the unit sphere $S^d$:  
\begin{enumerate}
\item[$\bullet$] a unit speed geodesic starting from $x\in S^d$ takes the form 
\begin{eqnarray}
\gamma(t)=x\,\cos(t)+v\,\sin(t),
\label{eqn0b}
\end{eqnarray}
where $v\in\mathbb{R}^{d+1}$ such that $\|v\|=1$ and $\langle x,v\rangle=0$; 
\item[$\bullet$] the exponential map at $x$ has the expression 
\begin{eqnarray}
\exp_x(v)=x\,\cos(\|v\|)+\frac{v}{\|v\|}\,\sin(\|v\|)
\label{eqn0d}
\end{eqnarray}
for any $v\in\mathbb{R}^{d+1}\setminus\{0\}$ such that $\langle x,v\rangle=0$; 
\item[$\bullet$] the inverse exponential at $x$ is given by
\[\exp_x^{-1}(x')=\frac{x'-\langle x,x'\rangle x}{\|x'-\langle x,x'\rangle x\|}\arccos(\langle x,x'\rangle)\]
for all $x'\in S^d\setminus\{-x\}$; 
\item[$\bullet$] the parallel transport of the tangent vector $w\in\tau_{\gamma(t_0)}(S^d)$ along the geodesic $\gamma$ defined by \eqref{eqn0b} is the vector field $w(t)$ along $\gamma(t)$ given by
\[w(t)=w-\langle w,v\rangle(v-\dot\gamma(t)),\]
where $v=\dot\gamma(0)$, as $w(t)\in\tau_{\gamma(t)}(S^d)$ and as the covariant derivative of $w(t)$ along $\gamma$ is $\nabla_{\dot\gamma(t)}w(t)=\dot w(t)-\langle\dot w(t),\gamma(t)\rangle\gamma(t)$.
\end{enumerate} 

See Figure \ref{tanfig} for illustrations of these concepts for $S^2$.

\section{Smoothing splines in shape space}\label{Sec3}

The method in the previous section provides a general framework for fitting smoothing splines on manifolds, and we now investigate the particular case
of fitting smoothing splines in shape space. 
Recall that, for a configuration in $\mathbb{R}^m$ with $k(>m)$ labelled landmarks, its pre-shape is what is left after the effects of translation and scaling are removed and that this pre-shape can be represented by an $m\times(k-1)$ matrix $X$ with  $\|X\|^2=\hbox{tr}(XX^\top)=1$ 
\citep{Kendall84,KBCL}. The space $\mathcal{S}_m^k$ of such pre-shapes is known as the pre-shape space of configurations of $k$ labelled landmarks in $\mathbb{R}^m$ and is the unit sphere of dimension $m(k-1)-1$, i.e. $\mathcal{S}_m^k=S^{m(k-1)-1}$. The Kendall shape space $\Sigma^k_m$ of configurations with $k$ labelled landmarks in $\mathbb{R}^m$ is then the quotient space of $\mathcal{S}^k_m$ by the rotation group $SO(m)$ acting on the left \citep{KBCL}. We shall use $[X]$ to denote the shape of the pre-shape $X$.

Unfortunately, due to their non-trivial geometric structure, the direct implementation on the Kendall shape spaces for configurations in $\mathbb{R}^m$ ($m\geqslant3$) of the exponential map and, in particular, of the parallel transport turns out to be challenging \citep{KBCL}. {This in turn makes the practical implementation of the spline fitting idea proposed in the previous section difficult.} One possible way to overcome this difficulty is to explore the fact that Kendall shape spaces are the quotient spaces of Euclidean spheres and to use the much simpler structure of spheres, as has previously been done by many in various different statistical investigations in shape analysis. For example, \cite{Le03} obtained an equivalent, qualitative description of the parallel transport on shape spaces via the pre-shape sphere. In the case of shapes of configurations in $\mathbb{R}^2$, this description leads successfully to a closed expression for such an equivalent notion on the pre-shape sphere, which has then been used for statistical analysis of shape in \cite{KDL}. Unfortunately, as pointed out in \cite{Le03}, such a closed expression is generally unavailable on the shape space of configurations in $\mathbb{R}^m$ for $m\geqslant3$.

In Section \ref{algorithm}, we will first provide an overview of the shape spline fitting algorithm. The full technical details of the algorithm require 
further developments, and in Section \ref{background} first we summarize facts about the tangent space and geodesics on the shape space of labelled 
configurations in $\mathbb{R}^m$. Then in Section \ref{ParallelTransport} we extend \cite{Le03}'s result which allows us to carry out parallel transport for such shape 
spaces, and hence enables us to implement the shape smoothing splines. We conclude the section with some results for size-and-shape space. 

\subsection{Shape spline fitting algorithm}\label{algorithm}
We aim to construct an algorithm to calculate the shape-space spline, or simply the shape spline, for a given set of configurations in $\mathbb{R}^m, m \ge 3$. 
The basic idea of the spline fitting algorithm is similar to that outlined by \cite{JK} for a sphere and \cite{KDL} for planar shape space. 
We start with a base path consisting of a piecewise geodesic path through the objects on the manifold. 
We then unroll and unwrap the data with respect to the base path to the tangent space at the start point of the base path.
We fit a smoothing spline in that tangent space and then wrap the spline interpolated points back to the manifold, which are used as the new piecewise geodesic base path in the next iteration. 
The procedure is repeated until there is little difference between the new and old base paths.  
An overview of the shape spline fitting algorithm is given here:

\begin{algorithm} 
\begin{enumerate}
$\;$ 
\item For a time series of $n+1$ configurations carry out registration of successive pre-shapes using ordinary Procrustes analysis to remove relative rotation information. 
\item Construct $G$ time points between each pair of observed times, which will be used for interpolation.
\item Take the piecewise horizontal geodesic path through the pre-shapes as an initial base path.
\item Unroll and unwrap the data with respect to the base path to the horizontal tangent space at the start point of the base path.
\item Fit a cubic smoothing spline to the data in that horizontal tangent space, with smoothing parameter $\lambda$.
\item Wrap the interpolated data back to the pre-shape sphere (consisting of $(G+1)n+1$ interpolated objects).
\item Register the interpolated objects using Procrustes analysis. The piecewise geodesic through these objects is the new base path.
\item If the maximum distance between the interpolated objects on the old and new base bath is greater than a threshold $\varepsilon>0$, then register the original objects using Procrustes analysis and go to 4.
Otherwise stop. 
\end{enumerate}
\end{algorithm}

We provide a more detailed and technical description of the algorithm in the Appendix, which requires the technical developments in Sections \ref{background} and \ref{ParallelTransport}. 

The linear spline is fitted to the data using the same iterative algorithm except that in step 5 we fit a linear spline rather than the cubic smoothing spline (\ref{cubicspline}). Also, for the 
single geodesic case ($\lambda \to \infty$) again the same iterative algorithm is used. 

{Note that for $m\geqslant3$ and $k\geqslant m+1$, while Kendall shape space $\Sigma_m^k$ is a complete metric space, as a Riemannian manifold it has a singular subset comprising the shapes whose pre-shape matrices have ranks at most $m-2$. Thus, if some data points are close to the singularity set, some extra care may be required in the above step 6 to check whether the ranks of any of the data matrices at any stage are less than, or equal to, $m-2$. If this happens, the proposed algorithm fails. Nevertheless, the singularity set has a high co-dimension and so is `tiny'. Hence, the possibility for this to happen is small in applications.}

Our data analyses have been implemented in \texttt{R} \citep{RDevel} and routines are available at
\begin{center}
\begin{small}
{\tt https://www.maths.nottingham.ac.uk/plp/pmzild/papers/splines3D}
\end{small}
\end{center}
The function {\tt smooth.spline} is used for fitting the cubic smoothing splines; the function {\tt elspline} in the package {\tt lspline} \citep{lspline17} is
used for fitting linear splines; 
 and the package {\tt shapes} is used for the Procrustes analysis and the relevant shape functions \citep{Dryden-shapes}.

The choice of the smoothing parameter $\lambda$, adapted to the data, can be determined by applying the (Euclidean) leave-one-out cross-validation method (\cite{ET93}, Chapter 17) to the unrolled shape spline and the unwrapped shape data. That is, we chose the optimal $\lambda$ which minimizes   
\begin{eqnarray*}\label{eq:cv}
CV(\lambda)=\frac{1}{n}\sum_{i = 1}^n \left\| v_i - \widehat{f_{\lambda, -i}^\dagger} (t_i)\right\|^2,
\end{eqnarray*}
where $\widehat{f_{\lambda, -i}^\dagger}$ is the unrolling of $\widehat f_{\lambda,-i}$ to the tangent space at its starting point, $\widehat f_{\lambda,-i}$ is the shape smoothing spline with the $i$th observation excluded and $v_i$ is the unwrapped $i$th observed shape data with respect to $\widehat f_{\lambda,-i}$.
 In several examples that we have considered, the spline fitting algorithm 
converges for the $\lambda$ chosen by cross-validation. 
However, the algorithm may not converge on some occasions, particularly if an inappropriate choice of $\lambda$ is made, and so setting an upper limit on the number of 
iterations (e.g. 20) is helpful in practice.

\subsection{Shape space, tangent space and shape geodesics}\label{background}
Writing $M(m,k-1)$ for the space of $m\times(k-1)$ real matrices, then the tangent space to the pre-shape sphere $\mathcal{S}_m^k$ at $X\in\mathcal{S}_m^k$ is
\[\tau^{\phantom{A}}_X(\mathcal{S}_m^k)=\{V\in M(m,k-1)\,:\,\hbox{ tr}(XV^\top)=0\}.\]
The horizontal subspace, which is the same as the Procrustean tangent space, of $\tau^{\phantom{A}}_X(\mathcal{S}_m^k)$, with respect to the quotient map from $\mathcal{S}_m^k$ to the shape space $\Sigma_m^k$, can be expressed as
\[\H_X(\S_m^k) =\{V\in M(m,k-1)\,:\,\tr(XV^\top) = 0, XV^\top \mbox{ is symmetric} \}.\]
The horizontal subspace $\H_X(\S_m^k)$ is isometric to the tangent space $\tau^{\phantom{A}}_{[X]}(\Sigma_m^k)$ to $\Sigma_m^k$ at the shape $[X]$ of $X$, and this gives a useful isometric representation of the tangent space of $\Sigma_m^k$ at $[X]$ for the statistical analysis of shapes \citep{KBCL}. {Note that the shape of a configuration becomes singular when the rank of the configuration matrix is less than $m-1$. To apply the results and methodology of the previous section, we have in fact implicitly restricted ourselves to situations where shapes are non-singular as is the case for most applications. We shall continue to make this restriction. Since the set of singular shapes has a high co-dimension and geodesics between any two non-singular shapes never pass through singular shapes, this restriction is relatively mild.}
 
 A unit-speed geodesic $\gamma$ in the shape space starting from $[X]$ can be isometrically represented by a so-called horizontal geodesic in $\S_m^k$ of the form
\[\gamma(s)=X\cos s+ V\sin s,\qquad s\in[0,\pi/2),\]
where $V\in\H_X(\S_m^k)$ and $\|V\|=1$. The path $\gamma$ is usually referred to as the horizontal lift of $\gamma$. To obtain a representation of a shortest geodesic between two shapes $[X_1]$ and $[X_2]$, we take the two pre-shapes $X_1$ and $X_2$ of $[X_1]$ and $[X_2]$ respectively such that $X_1X_2^\top$ is symmetric and all its eigenvalues are non-negative except possibly for $\lambda_m$, the smallest one, where $\hbox{sign}(\lambda_m)=\hbox{sign}(\hbox{det}(X_1X_2^\top))$. Then, a unit-speed shortest geodesic from $[X_1]$ to $[X_2]$ can be isometrically represented by the horizontal geodesic from $X_1$ to $X_2$: 
\begin{eqnarray}
\gamma(s)=X_1 \cos s + V_{X_1,X_2} \sin s,\qquad s \in [0,s_0],
\label{eqn0}
\end{eqnarray}
where
\[V_{X_1,X_2}=\frac{1}{\sin s_0}\{X_2 - X_1 \cos s_0\}\in\H_{X_1}\]
and $s_0 = \rho([X_1], [X_2])$ is the Riemannian shape distance between $[X_1]$ and $[X_2]$. Note that $\|V_{X_1,X_2}\|=1$.

Thus, it follows from the expression \eqref{eqn0} for the horizontal geodesic from $X_1$ to $X_2$ that, if there is a unique geodesic between $[X_1]$ and $[X_2]$, the inverse exponential map $\exp_{[X_1]}^{-1}([X_2])$ on the shape space $\Sigma_m^k$ can be isometrically represented by its horizontal lift on $\S_m^k$ given by
\begin{eqnarray}
\exp_{X_1}^{-1}(X_2)=s_0\,V_{X_1,X_2}.
\label{eqn0c}
\end{eqnarray}

\subsection{Parallel transport}\label{ParallelTransport}
Turning to parallel transport, it is shown in \cite{Le03} that, along a horizontal geodesic $\gamma(s)$ in $\S_m^k$, a vector field $V(s)$ is horizontal and its projection to $\tau^{\phantom{A}}_{[\gamma(s)]}(\Sigma_m^k)$ is the parallel transport, along the shape geodesic $[\gamma(s)]$, of the projection of $V(0)$ to $\tau^{\phantom{A}}_{[\gamma(0)]}(\Sigma_m^k)$ if and only if $V(s)$ satisfies the following three conditions:
\begin{eqnarray}
&&\hbox{tr}(V(s)\,\gamma(s)^\top)=0\label{eqn3}\\
&&V(s)\,\gamma(s)^\top=\gamma(s)\,V(s)^\top\label{eqn4}\\
&&\dot V(s)-\hbox{tr}(\gamma(s)\,\dot V(s)^\top)\,\gamma(s)=A(s)\,\gamma(s)\quad\hbox{ for some $A(s)$ such} \label{eqn5}\\
&&\qquad\qquad\qquad\qquad\qquad\qquad\qquad\qquad\qquad\hbox{ that }A(s)=-A(s)^\top.\nonumber
\end{eqnarray}
{Thus, the solution $V(s)$ to \eqref{eqn3}, \eqref{eqn4} and \eqref{eqn5} provides an isometric representation, on $\mathcal S_m^k$, of the parallel transport $P^{[\gamma(0)]}_{[\gamma(s_0)]}$ along $[\gamma(s)]$, $0\leqslant s\leqslant s_0$:
\begin{eqnarray}
\Psi_{\gamma(s_0)}^{\gamma(0)}:\,V(0)\mapsto V(s_0).
\label{eqn5a}
\end{eqnarray}
}
The usefulness of this result for practical implementation lies crucially in obtaining a more quantitative description of the skew-symmetric matrix $A(s)$ in \eqref{eqn5}, which {we derive with} the following result.  

\begin{proposition}\label{Prop2}
Let $\gamma(s)$, $0\leqslant s\leqslant s_0$, be a given horizontal $\mathcal{C}^1$-curve in ${\mathcal S}_m^k$ and $V$ be a given horizontal tangent vector in $\tau^{\phantom{A}}_{\gamma(0)}({\mathcal S}_m^k)$. Assume that the rank of $\gamma(s)$ is at least $m-1$, except for at most finitely many $s$. Then, the vector field $V(s)$ along $\gamma(s)$ is horizontal and the projection of $V(s)$ to $\tau^{\phantom{A}}_{[\gamma(s)]}(\Sigma_m^k)$ is the parallel transport, along the shape curve $[\gamma(s)]$, of the projection of $V$ to $\tau^{\phantom{A}}_{[\gamma(0)]}(\Sigma_m^k)$ if and only if $V(s)$ is the solution of
\begin{eqnarray}
\begin{array}{rcl}
\dot V(s)&=&-{\rm tr}(\dot\gamma(s)\,V(s)^\top)\,\gamma(s)+A(s)\,\gamma(s),\qquad s\in[0,\,s_0],\\
V(0)&=&V,
\end{array}
\label{eqn1}
\end{eqnarray}
where $A(s)$ is skew-symmetric and is the unique solution to
\begin{eqnarray}
A(s)\,\gamma(s)\,\gamma(s)^\top+\gamma(s)\,\gamma(s)^\top A(s)=\dot\gamma(s)\,V(s)^\top-V(s)\,\dot\gamma(s)^\top.
\label{eqn2}
\end{eqnarray}
\label{prop1}
\end{proposition}

\begin{proof}
Note first that $V\in\mathcal{H}_{\gamma(0)}({\mathcal S}_m^k)$ implies that $\hbox{tr}(V\gamma(0)^\top)=0$ and $V\gamma(0)^\top$ is a symmetric matrix.

Assume that $V(s)$ satisfies \eqref{eqn3}, \eqref{eqn4} and \eqref{eqn5}. Since tr$(V(0)\,\gamma(0)^\top)=0$, condition \eqref{eqn3} holds for all $s\in[0,\,s_0]$ if and only if
\begin{eqnarray}
\hbox{tr}(\dot V(s)\,\gamma(s)^\top)+\hbox{tr}(V(s)\,\dot\gamma(s)^\top)=0,\qquad s\in[0,\,s_0].
\label{eqn6}
\end{eqnarray}
Thus, under conditions \eqref{eqn3} and \eqref{eqn5}, $V(s)$ satisfies \eqref{eqn1} for some skew-symmetric matrix $A(s)$. Since $V(0)\,\gamma(0)^\top=\gamma(0)\,V(0)^\top$, condition \eqref{eqn4} holds if and only if
\[\dot V(s)\,\gamma(s)^\top+V(s)\,\dot\gamma(s)^\top=\dot\gamma(s)\,V(s)^\top+\gamma(s)\,\dot V(s)^\top,\]
which is equivalent to
\begin{eqnarray}
\dot V(s)\,\gamma(s)^\top-\gamma(s)\,\dot V(s)^\top=\dot\gamma(s)\,V(s)^\top-V(s)\,\dot\gamma(s)^\top.
\label{eqn7}
\end{eqnarray}
Hence, if \eqref{eqn5} holds for some $A(s)$ such that $A(s)=-A(s)^\top$, $A(s)$ must satisfy
\[\dot V(s)\,\gamma(s)^\top-\gamma(s)\,\dot V(s)^\top=A(s)\,\gamma(s)\,\gamma(s)^\top+\gamma(s)\,\gamma(s)^\top A(s).\]
Thus, conditions \eqref{eqn4} and \eqref{eqn5} together imply that $A(s)$ must satisfy \eqref{eqn2}. 

The uniqueness of the solution to \eqref{eqn2} follows from the fact that, for a given $m\times m$ symmetric non-negative definite matrix $S=R\Lambda R^\top$ of rank at least $m-1$, where $R\in{\bf O}(m)$ and $\Lambda=\hbox{diag}\{\lambda_1,\ldots,\lambda_m\}$, and a given skew-symmetric matrix $B$, there is a unique skew-symmetric matrix $A$ satisfying the equation $AS+SA=B$. To see the latter, write $\tilde A=(\tilde a_{ij})=R^\top AR$ and $\tilde B=(\tilde b_{ij})=R^\top BR$. Then, $AS+SA=B$ if and only if $\tilde A\Lambda+\Lambda\tilde A=\tilde B$, which is if and only if $\lambda_i\tilde a_{ij}+\lambda_j\tilde a_{ij}=\tilde b_{ij}$ for $i<j$, i.e. if and only if $\tilde a_{ij}=\tilde b_{ij}/(\lambda_i+\lambda_j)$. 

On the other hand, if $V(s)$ is the solution to \eqref{eqn1} with $A(s)$ being determined by \eqref{eqn2}, then $V(s)$ satisfies \eqref{eqn7} and so \eqref{eqn4} holds. However, tr$(A(s)\,\gamma(s)\,\gamma(s)^\top)=\hbox{tr}\left(\gamma(s)\,\gamma(s)^\top A(s)\right)$, so that condition \eqref{eqn2} implies that we must have 
\[\hbox{tr}\left(A(s)\,\gamma(s)\,\gamma(s)^\top\right)=0.\]
This, together with \eqref{eqn7}, shows that $V(s)$ also satisfies \eqref{eqn6}. Hence, both \eqref{eqn3} and \eqref{eqn5} hold.
\end{proof}

Note that $V(s)$ given in Proposition \ref{prop1} is generally not the parallel transport of $V$ along $\gamma(s)$ in ${\mathcal S}_m^k$. The reason for this is that, in terms of the covariant derivative $\nabla$ on ${\mathcal S}_m^k$, one requires $\nabla_{\dot\gamma(s)}V(s)=0$ for $V(s)$ to be the parallel transport of $V(0)$ along $\gamma$ in ${\mathcal S}_m^k$, while one only needs $\nabla_{\dot\gamma(s)}V(s)$ to be orthogonal to the horizontal subspace ${\mathcal H}_{\gamma(s)}({\mathcal S}_m^k)$, as explained in \cite{Le03}, for the projection of $V(s)$ to be the parallel transport along $[\gamma(s)]$ of the projection of $V(0)$. The latter is the one that we require.

To apply the above results to the parallel transport of the projection of $V(s_0)$ along $[\gamma(s)]$, $0\leqslant s\leqslant s_0$, back to $\tau^{\phantom{A}}_{[\gamma(0)]}(\Sigma_m^k)$ that we require for the unrolling and unwrapping procedures, we re-parametrize the curve $\gamma(s)$ and the vector field $V(s)$ to $\tilde\gamma(s)=\gamma(s_0-s)$ and $\tilde V(s)=V(s_0-s)$, say, respectively. With this modification in mind,  the result of Proposition \ref{prop1}{, in particular  \eqref{eqn1}, may} 
appear to provide a numerical method to approximate
$\Psi^{\gamma(s_0)}_{\gamma(0)}$  in the usual way: with step size $\delta=s_0/\ell$,  $V^*(\ell\delta)=V(s_0)$,

\begin{eqnarray}
V^*((i-1)\delta)=V^*(i\delta)-\hbox{tr}(\dot\gamma(i\delta)V^*(i\delta)^\top)\,\gamma(i\delta)\,\delta+A(i\delta)\,\gamma(i\delta)\,\delta,
\label{eqn5b}
\end{eqnarray}
where $A(i\delta)$ is updated at each step using \eqref{eqn2}. 
{However,} $V^*((i-1)\delta)$ obtained in this way is generally neither tangent to ${\mathcal S}_m^k$ nor horizontal. Hence, 
{the use of \eqref{eqn5b}} requires finding, at each step, the re-normalized horizontal component of $V^*((i-1)\delta)$. That is, at each step, we need to project $V^*((i-1)\delta)$ to $\tau^{\phantom{A}}_{\gamma(({i-1})\delta)}({\mathcal S}_m^k)$ and then to the horizontal subtangent space of ${\mathcal S}_m^k$ at $\gamma(({i-1})\delta)$, followed by normalizing it so that its length remains equal to that of $V(s_0)$. 

Now, if $X$ is a pre-shape matrix with rank$(X)\geqslant m-1$, the projection of $V\in M(m,k-1)$ onto the tangent space $\tau^{\phantom{A}}_X(\S_m^k)$ is given by
\begin{eqnarray}
V_t=V-\hbox{tr}(VX^\top)\,X.
\label{eqn5c}
\end{eqnarray}
To obtain the projection of $V_t\in\tau^{\phantom{A}}_X(\S_m^k)$ onto $\mathcal{H}_X(\S_m^k)$, we note that the orthogonal complementary subspace to $\mathcal{H}_X(\S_m^k)$ in $\tau^{\phantom{A}}_X(\S_m^k)$ is $\mathcal{V}_X=\{AX\,:\, A=-A^\top\}$. 
Since $\{A_{ij}\,:\, 1\leqslant i<j\leqslant m\}$ forms an orthonormal basis for the space of $m\times m$ skew-symmetric matrices where $A_{ij}=\frac{1}{\sqrt{2}}\{E_{ij}-E_{ji}\}$ and $E_{ij}$ is the $m\times m$ square matrix with all its elements zero except for the $(i,j)$th element which is one, $\{A_{ij}X/\|A_{ij}X\|\,:\,1\leqslant i<j\leqslant m\}$ forms an orthonormal basis for $\mathcal{V}_X$. Thus, the projection of $V_t\in\tau^{\phantom{A}}_X(\S_m^k)$ onto $\mathcal{H}_X(\S_m^k)$ is given by
\begin{eqnarray}
V_h=V_t-\sum_{1\leqslant i<j\leqslant m}\frac{1}{\|A_{ij}X\|^2}\hbox{tr}\left(A_{ij}X {V_t}^\top\right)A_{ij}X.
\label{eqn5d}
\end{eqnarray}

Combining \eqref{eqn5b}, \eqref{eqn5c} and \eqref{eqn5d}, we have an approximation to $\Psi_{\gamma(s_0)}^{\gamma((\ell-1)\delta)}(V(s_0))$ 
of (\ref{eqn5a}) along the horizontal geodesic $\gamma(s)$, after the first $\delta$-step, as
\begin{eqnarray}
\Psi^{\gamma(s_0)}_{\gamma((\ell-1)\delta)}(V(s_0))\approx\|V(s_0)\|\,\frac{W}{\|W\|}=V((\ell-1)\delta),
\label{eqn5e}
\end{eqnarray}
where $W$ is obtained using \eqref{eqn5d} with $V_t$ there replaced by
\begin{eqnarray}
V^*((\ell-1)\delta)-\hbox{tr}(V^*((\ell-1)\delta)\,\gamma((\ell-1)\delta)^\top)\,\gamma((\ell-1)\delta)
\label{eqn5f}
\end{eqnarray}
and with $X$ there replaced by $\gamma((\ell-1)\delta)$, and where $V^*((\ell-1)\delta)$ in \eqref{eqn5f} is obtained using \eqref{eqn5b}. Note that \eqref{eqn5f} is the projection of $V^*((\ell-1)\delta)$ to $\tau^{\phantom{A}}_{\gamma((\ell-1)\delta)}(\S_m^k)$ as given by \eqref{eqn5c}. Recall that $s_0=\ell\delta$. Then, since 
\begin{eqnarray*}
\Psi^{\gamma(s_0)}_{\gamma((\ell-2)\delta)}(V(s_0))&=&\Psi_{\gamma((\ell-2)\delta)}^{\gamma((\ell-1)\delta)}\circ\Psi^{\gamma(s_0)}_{\gamma((\ell-1)\delta)}(V(s_0))\\
&\approx&\Psi_{\gamma((\ell-2)\delta)}^{\gamma((\ell-1)\delta)}(V((\ell-1)\delta)),
\end{eqnarray*}
application of \eqref{eqn5e} with $\ell$ replaced by $\ell-1$ again gives us an approximation to\\
$\Psi^{\gamma(s_0)}_{\gamma((\ell-2)\delta)}(V(s_0))$. 
Finally, repeating this backward induction on $\ell$ will result in an approximation to $\Psi^{\gamma(s_0)}_{\gamma(0)}(V(s_0))$. A similar procedure forward on $\ell$ will give an approximation to $\Psi^{\gamma(0)}_{\gamma(s_0)}(V(0))$.


\subsection{Size-and-shape splines}\label{sssplines}
All the results for shape spaces have analogues in the size-and-shape setting, where we do not have invariance under scaling. For example, for a configuration in $\mathbb{R}^m$ with $k(>m)$ labelled landmarks, its pre-size-and-shape is what is left after the effects of translation are removed. This pre-size-and-shape can be represented by an $m\times(k-1)$ matrix $\tilde X\in M(m,k-1)$ and the space of the pre-size-and-shapes, $\mathcal{SS}_m^k$, is identical with $M(m,k-1)$. For $\tilde X\not=0$, the tangent space $\tau^{\phantom{A}}_{\tilde X}(\mathcal{SS}_m^k)$ is then also $M(m,k-1)$ and the horizontal subspace of $\tau^{\phantom{A}}_{\tilde X}(\mathcal{SS}_m^k)$, with respect to the quotient map from $\mathcal{SS}_m^k$ to the Kendall size-and-shape space $S\Sigma_m^k$, is given by
\[\H_{\tilde X}(\mathcal{SS}_m^k) =\{V\in M(m,k-1)\mid \tilde XV^\top \mbox{ is symmetric} \}.\]
Horizontal geodesics in $\mathcal{SS}_m^k$ starting from $\tilde X$ take the form
\[\tilde\gamma(s)=\tilde X+ s\,V,\]
where $V\in\H_{\tilde X}(\mathcal{SS}_m^k)$. For two given size-and-shapes $[\tilde X_1]$ and $[\tilde X_2]$, let $\tilde X_1$ and $\tilde X_2$ be the pre-size-and-shapes of $[\tilde X_1]$ and $[\tilde X_2]$ respectively such that $\tilde X_1\tilde X_2^\top$ is symmetric and all its eigenvalues are non-negative except possibly for $\lambda_m$, the smallest one, where $\hbox{sign}(\lambda_m)=\hbox{sign}(\hbox{det}(\tilde X_1\tilde X_2^\top))$. Then, $\tilde X_2-\tilde X_1\in\mathcal{H}_{\tilde X}(\mathcal{SS}_m^k)$ and a shortest geodesic from $[\tilde X_1]$ to $[\tilde X_2]$ can be represented by the horizontal geodesic that connects $\tilde X_1$ and $\tilde X_2$: 
\begin{eqnarray}
\tilde \gamma(s)=\tilde X_1 + s\,(\tilde X_2-\tilde X_1),\qquad s \in [0,1].
\label{eqn0a}
\end{eqnarray}
Thus, the inverse exponential map $\exp_{[\tilde X_1]}^{-1}([\tilde X_2])$ on the size-and-shape space $S\Sigma_m^k$, using its horizontal lift on $\mathcal{SS}_m^k$, is isometrically represented by
\[\exp_{\tilde X_1}^{-1}(\tilde X_2)=\tilde X_2-\tilde X_1,\]
as long as the geodesic between $[\tilde X_1]$ and $[\tilde X_2]$ is unique. 

Moreover, a modification of the argument given in \cite{Le03} shows that, along a horizontal size-and-shape geodesic $\tilde\gamma(s)$ in $\mathcal{SS}_m^k$, the vector field $V(s)$ is horizontal and its projection to $\tau^{\phantom{A}}_{[\tilde\gamma(s)]}(S\Sigma_m^k)$ is the parallel translation, along the size-and-shape geodesic $[\tilde\gamma(s)]$, of the projection of $V(0)$ onto $\tau^{\phantom{A}}_{[\tilde\gamma(0)]}(S\Sigma_m^k)$ if and only if $V(s)$ satisfies the conditions \eqref{eqn4} and 
\[{\dot V(s)=A(s)\,\tilde\gamma(s)\quad\hbox{ for some }A(s)\hbox{ such that }A(s)=-A(s)^\top.}\]
Thus, it follows from the proof for Proposition \ref{prop1} that the result of Proposition \ref{prop1} can also be generalized to size-and-shape space.

\begin{corollary}
Let $\tilde\gamma(s)$, $0\leqslant s\leqslant s_0$, be a given horizontal $\mathcal{C}^1$-curve in $\mathcal{SS}_m^k$ and $V$ be a given horizontal tangent vector in $\tau^{\phantom{A}}_{\tilde\gamma(0)}(\mathcal{SS}_m^k)$. Assume that $\hbox{rank}(\tilde\gamma(s))\geqslant m-1$, except for at most finitely many $s$. Then, the vector field $V(s)$ along $\tilde\gamma(s)$ is horizontal and the projection of $V(s)$ to $\tau^{\phantom{A}}_{[\tilde\gamma(s)]}(S\Sigma_m^k)$ is the parallel transport, along the size-and-shape curve $[\tilde\gamma(s)]$, of the projection of $V$ onto $\tau^{\phantom{A}}_{[\tilde\gamma(0)]}(S\Sigma_m^k)$ if and only if $V(s)$ is the solution of
\begin{eqnarray*}
\dot V(s)&=&A(s)\,\tilde\gamma(s),\qquad s\in[0,\,s_0],\\
V(0)&=&V,
\end{eqnarray*}
where $A(s)$ is skew-symmetric and is the unique solution to \eqref{eqn2} with $\gamma$ being replaced by $\tilde\gamma$.
\end{corollary}

A practical alternative to working directly in size-and-shape space is to work with independent splines in the product of the univariate centroid size space ($\mathbb R^+$) and shape space ($\Sigma^k_m$). In many practical applications it can
be helpful to separate size and shape in this manner, so that size does not overly dominate the statistical analysis.

\section{Application}\label{Sec4}
\subsection{Peptide data} 
In the biomedical sciences it is of great interest to study the changes in shape of a protein, as shape 
is an important component of a protein's function and hence its role in the cell. 
Broad aims in the study of protein dynamics are to investigate how proteins fold into a 3D 
functional shape  and to estimate the full range of such functional shapes for a given protein \citep[e.g.,see][]{Karplus83,Karplus05}. 
Protein folding is very sensitive to external processes and protein misfolding is a key component of various diseases.
 
Molecular dynamics simulations are often used to 
study the the possible configurations of molecules under various scenarios, and these are computationally intensive   
deterministic simulations which use Newtonian mechanics to model the movement of a molecule in a box surrounded by water molecules \citep[e.g.][]{AMBER13}. We consider a dataset of 
in the study of the alanine pentapeptide (Ala$_5$) which is a small protein (peptide) that consists of $k = 29$ atoms in $\R^3$. The data were provided by Professor Charles Laughton, School of Pharmacy, 
University of Nottingham, and further detail on this particular peptide is given by \cite{Margulisetal02} and \citet{Drydenetal19}. 
It is of interest to examine if there are preferred states, i.e. clusters of
shapes which are more commonly formed by the dynamic peptide. Also, low dimensional representations and the
patterns of temporal transitions between states are of interest. 

For this application a 
small temporal subsequence of 30 peptide configurations equally spaced in time is taken, spanning a 10 nanosecond ($10^{-8}s$) period, and we consider several different models for
predicting molecular shapes inbetween the observed peptides. 
Planar projections of these configurations are presented in Figures \ref{fig:peptide.fit}  and \ref{fig:peptide.fit2}  as rainbow coloured dots, 
where red and violet points indicate the first and last landmarks, respectively.
The peptide configuration at the start of this sequence is very irregular in all three dimensions at $t = 1$,
then it gradually straightens over time so that we see a more smoothly curved form at $t = 30$. 
{Their shapes vary substantially with largest pairwise Riemannian shape distance $1.028$ (where the maximum possible value is $\pi/2$) and  this motivates us to consider analysis based on unrolling and unwrapping to a base tangent space, rather than the more straightforward analysis based on a single tangent space projection.}

\begin{figure}[htbp]
\addtolength{\subfigcapskip}{-0.2in}
   \centering \subfigure[$t = 1$]
   {\includegraphics[scale = 0.3]{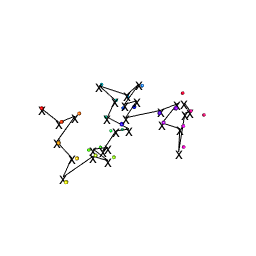}}
      \centering \subfigure[$t = 3$]
   {\includegraphics[scale = 0.3]{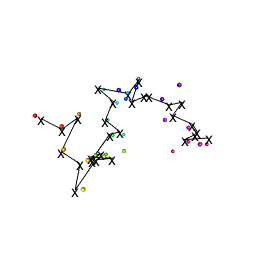}}
    \centering \subfigure[$t = 5$]
   {\includegraphics[scale = 0.3]{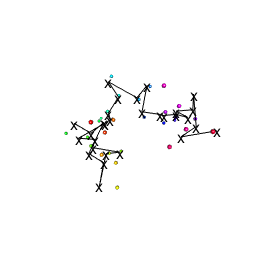}}
   \centering \subfigure[$t = 9$]
   {\includegraphics[scale = 0.3]{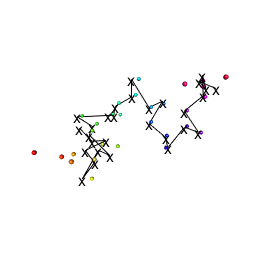}}
     \centering \subfigure[$t = 13$]
   {\includegraphics[scale = 0.3]{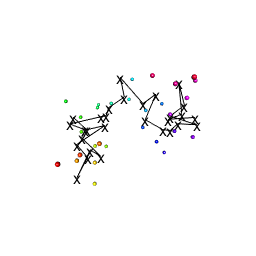}}
        \centering \subfigure[$t = 15$]
   {\includegraphics[scale = 0.3]{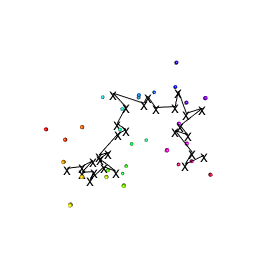}}
    \centering \subfigure[$t = 17$]
   {\includegraphics[scale = 0.3]{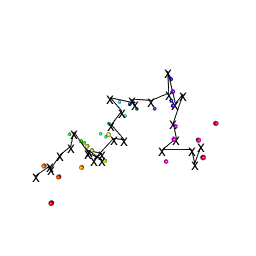}}
        \centering \subfigure[$t = 19$]
   {\includegraphics[scale = 0.3]{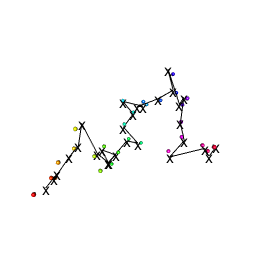}}
     \centering \subfigure[$t = 22$]
   {\includegraphics[scale = 0.3]{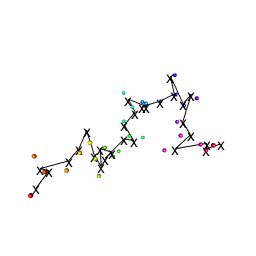}}
      \centering \subfigure[$t = 25$]
   {\includegraphics[scale = 0.3]{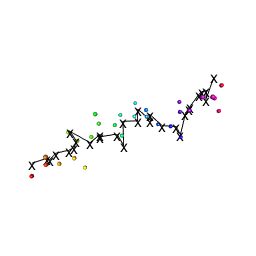}}
      \centering \subfigure[$t = 28$]
   {\includegraphics[scale = 0.3]{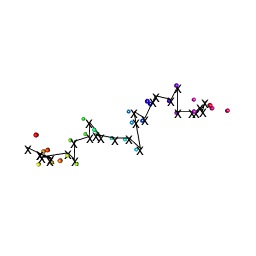}} 
   \centering \subfigure[$t = 30$]
   {\includegraphics[scale = 0.3]{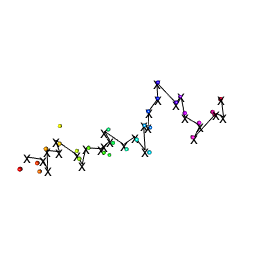}} 
   
     \caption{Configurations of a dynamic short peptide.
   In each subfigure, rainbow coloured dots indicate 29 landmarks and the connected points `$\times$' are the fitted configurations. One particular 
   view is given in (a)-(l).}
   \label{fig:peptide.fit}
\end{figure}

   \begin{figure}[htbp]
\addtolength{\subfigcapskip}{-0.2in}   
   \centering \subfigure[$t = 1$]
   {\includegraphics[scale = 0.3]{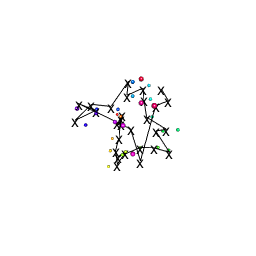}}
     \centering \subfigure[$t = 3$]
   {\includegraphics[scale = 0.3]{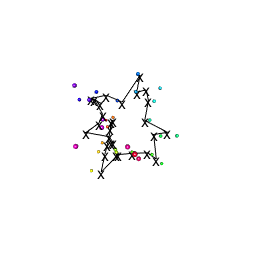}}
  \centering \subfigure[$t = 5$]
   {\includegraphics[scale = 0.3]{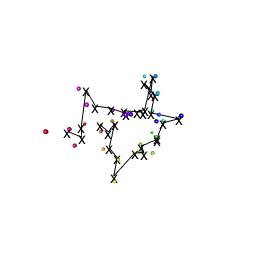}}
     \centering \subfigure[$t = 9$]
   {\includegraphics[scale = 0.3]{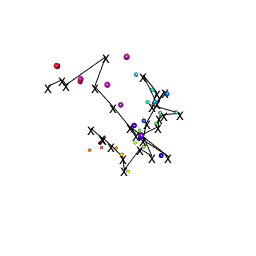}}
      \centering \subfigure[$t = 13$]
   {\includegraphics[scale = 0.3]{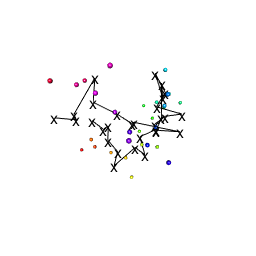}}
         \centering \subfigure[$t = 15$]
   {\includegraphics[scale = 0.3]{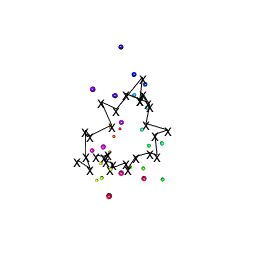}}
    \centering \subfigure[$t = 17$]
   {\includegraphics[scale = 0.3]{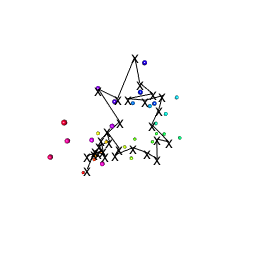}}
         \centering \subfigure[$t = 19$]
   {\includegraphics[scale = 0.3]{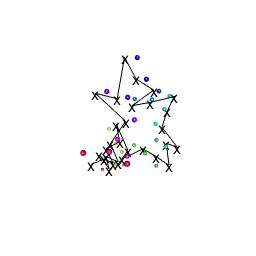}}
   \centering \subfigure[$t = 22$]
   {\includegraphics[scale = 0.3]{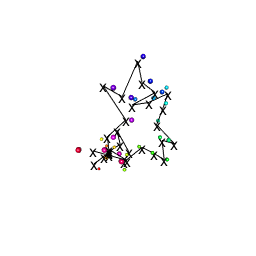}}
  \centering \subfigure[$t = 25$]
   {\includegraphics[scale = 0.3]{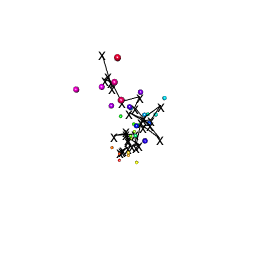}}
    \centering \subfigure[$t = 28$]
   {\includegraphics[scale = 0.3]{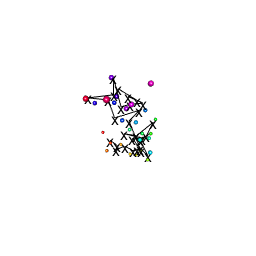}}
     \centering \subfigure[$t = 30$]
   {\includegraphics[scale = 0.3]{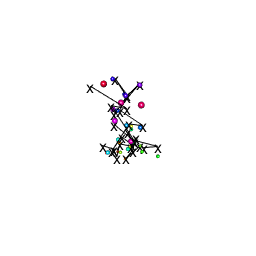}}   \caption{Configurations of a dynamic short peptide.
   In each subfigure, rainbow coloured dots indicate 29 landmarks and the connected points `$\times$' are the fitted configurations. 
   Orthogonal views to those of Figure \ref{fig:peptide.fit} are given in (a)-(l).}
   \label{fig:peptide.fit2}
\end{figure}

\subsection{Cubic shape smoothing spline fitting}
In Figures \ref{fig:peptide.fit} and \ref{fig:peptide.fit2} the fitted configurations using the cubic shape smoothing spline to the moving configurations are indicated by `$\times$' with connecting solid lines.  We used $G = 2$ interpolation points between each pair of data points in the base path and we chose $\varepsilon = 10^{-3}$ for the convergence in the fitting algorithm. 
The candidates for $\lambda$ in the cross-validation were taken from the set $\{10^{-9}, 10^{-7}, 10^{-6}, 10^{-5}, 10^{-4}, 10^{-3}, 10^{-1} \}$, and 
for the peptide data the chosen smoothing parameter is $\lambda=10^{-4}$. 

{Although the two dimensional projections of some peptides in Figure \ref{fig:peptide.fit} look close to collinearity, their actual three-dimensional figures are not close to collinearity as seen in Figure \ref{fig:peptide.fit2}. Hence, their shapes are not close to the singularity set. 
For example the 1st principal axis of the final peptide explains $90.5\%$ of the variability, and while this is 
fairly high it is quite far from being collinear ($\sim100\%$). 

We can see that the shape spline has indeed smoothed the path of configurations. In some instances the cubic smoothing spline fit is close 
to the data (e.g. $t=19$) and in other cases it is further away (e.g. $t=15$).

\subsection{Shape PCA}
For a lower dimensional view of the peptide shapes and the shape spline we can carry out principal components analysis of shape in the 
horizontal tangent space to pre-shape space \citep[Chapter 7]{Drydmard16}, which is implemented in R in the {\tt shapes} package \citep{Dryden-shapes}.  
Figure \ref{fig:peptide.pcscore} 
shows the first three shape principal component scores of the 30 peptide configurations in the tangent space at their mean shape. 
Principal component (PC) scores 1,2,3  explain 43.1\%,  22.0\% and 12.3\% of the shape variability respectfully, 
and so this plot includes the most important aspects of shape variability. The 
connected dots indicate the PC scores of the data points and the red lines on the other hand are the PC scores of the fitted paths, where their approximate starting points are indicated by `$S$' in all panels. The cubic shape smoothing spline is shown in Figure \ref{fig:peptide.pcscore}(a) and 
we can see that the smoothing spline path is a good fit to the data.  
\begin{figure}[htbp]
\begin{center}
    \subfigure[Smoothing cubic spline model]
   {\includegraphics[width=5cm]{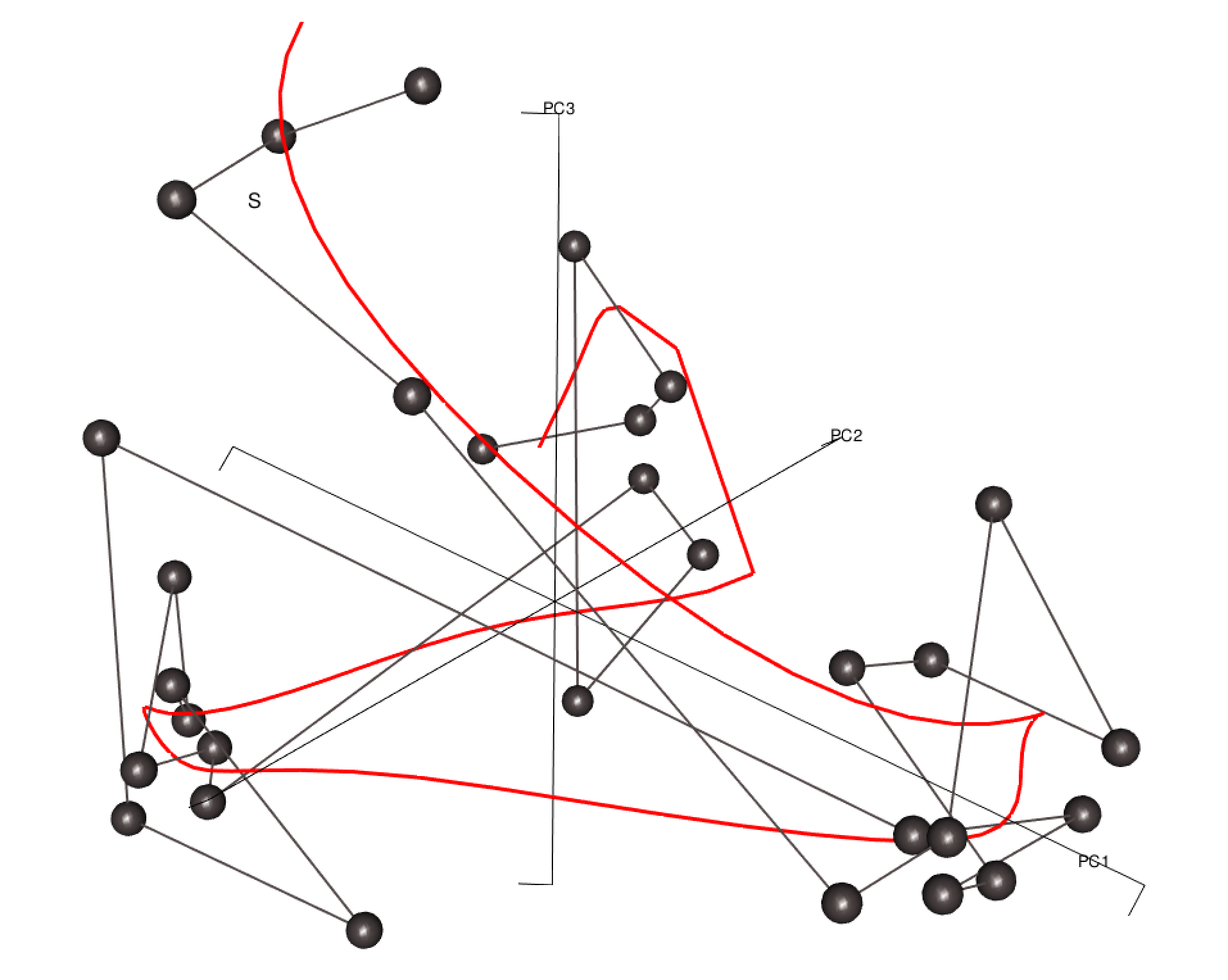}}
  \subfigure[One geodesic model]
   {\includegraphics[width=5cm]{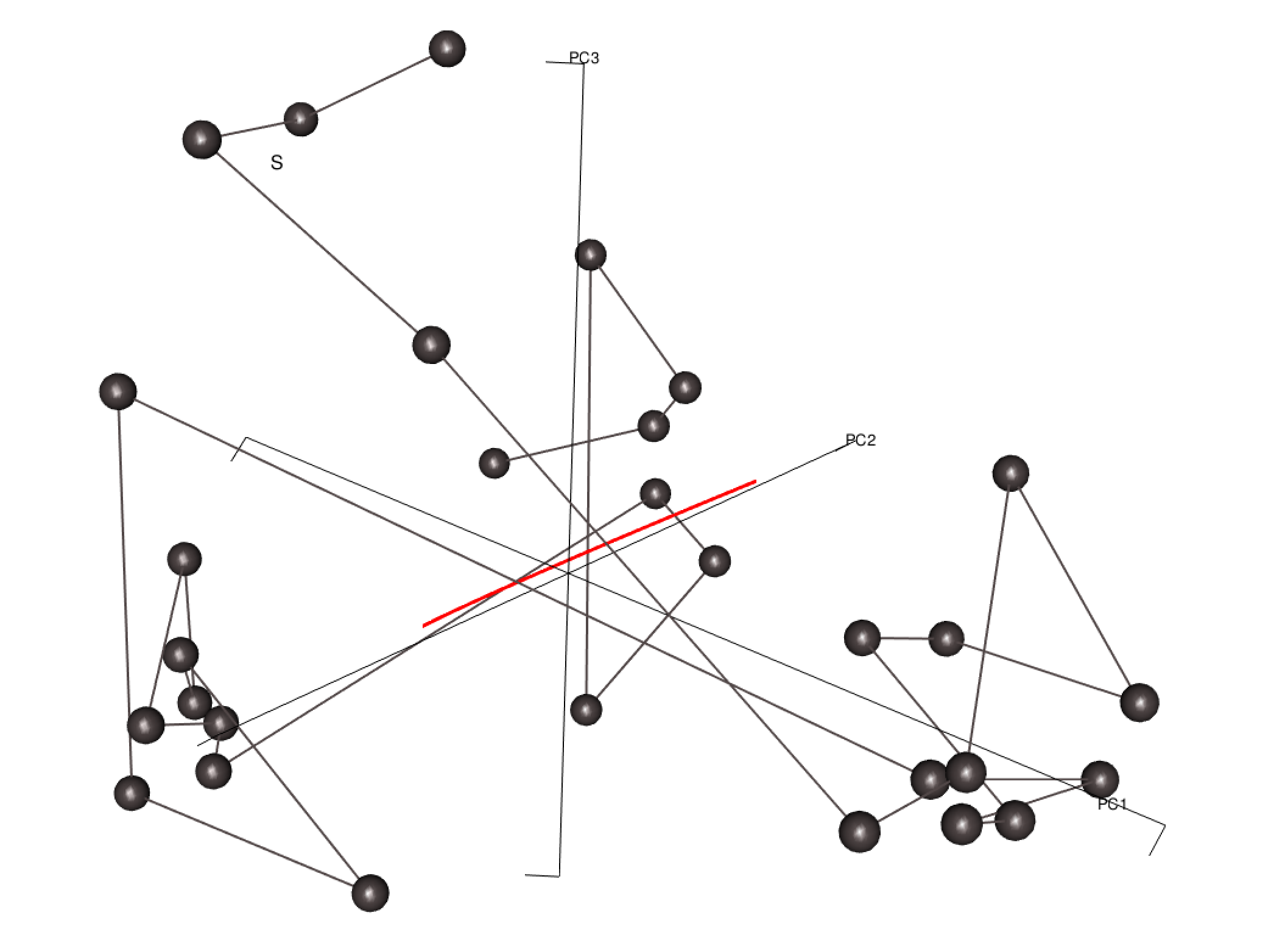}}\\
   \subfigure[Linear spline 3 knots model]
   {\includegraphics[width=5cm]{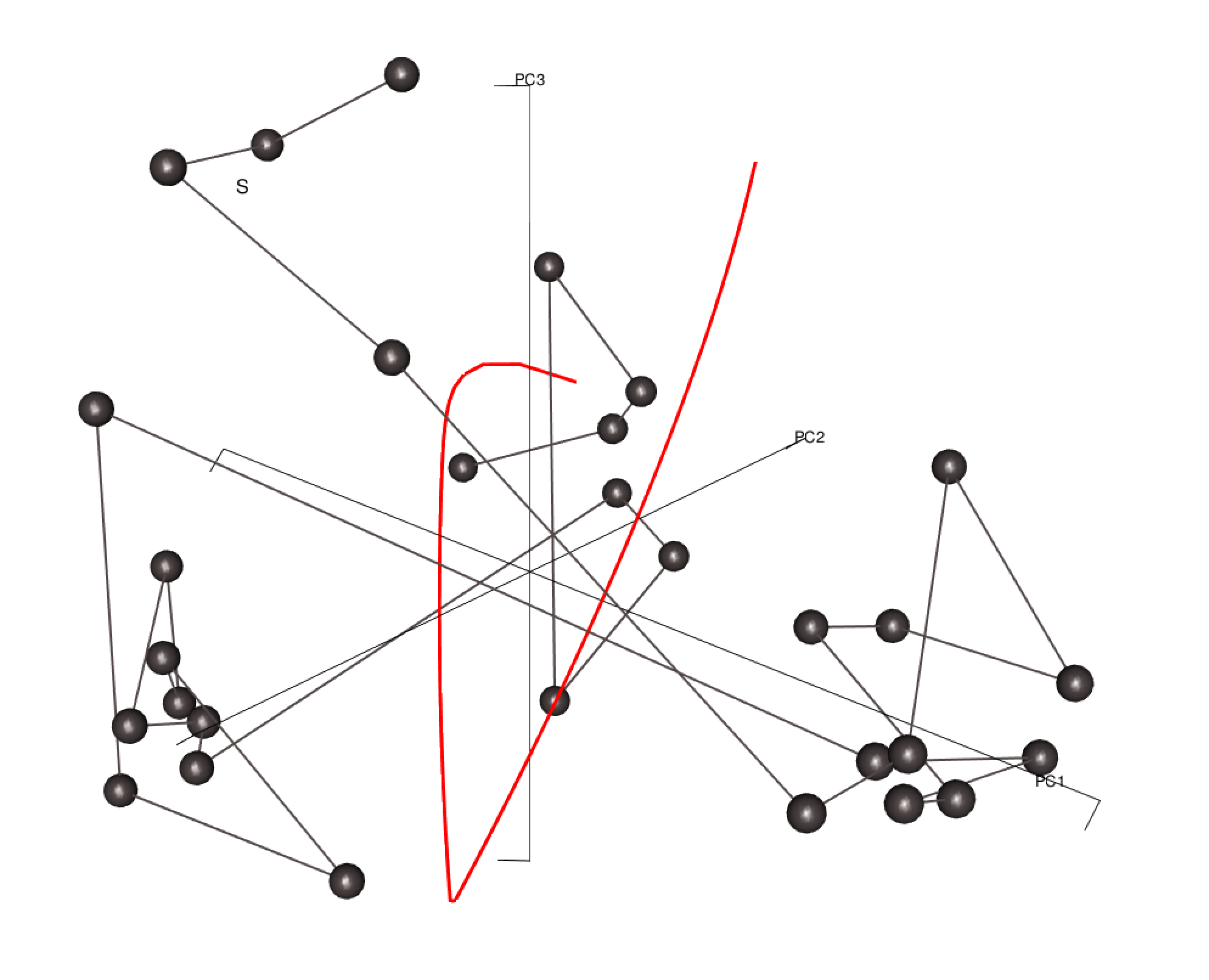}}
     \subfigure[Linear spline 4 knots model]
   {\includegraphics[width=5cm]{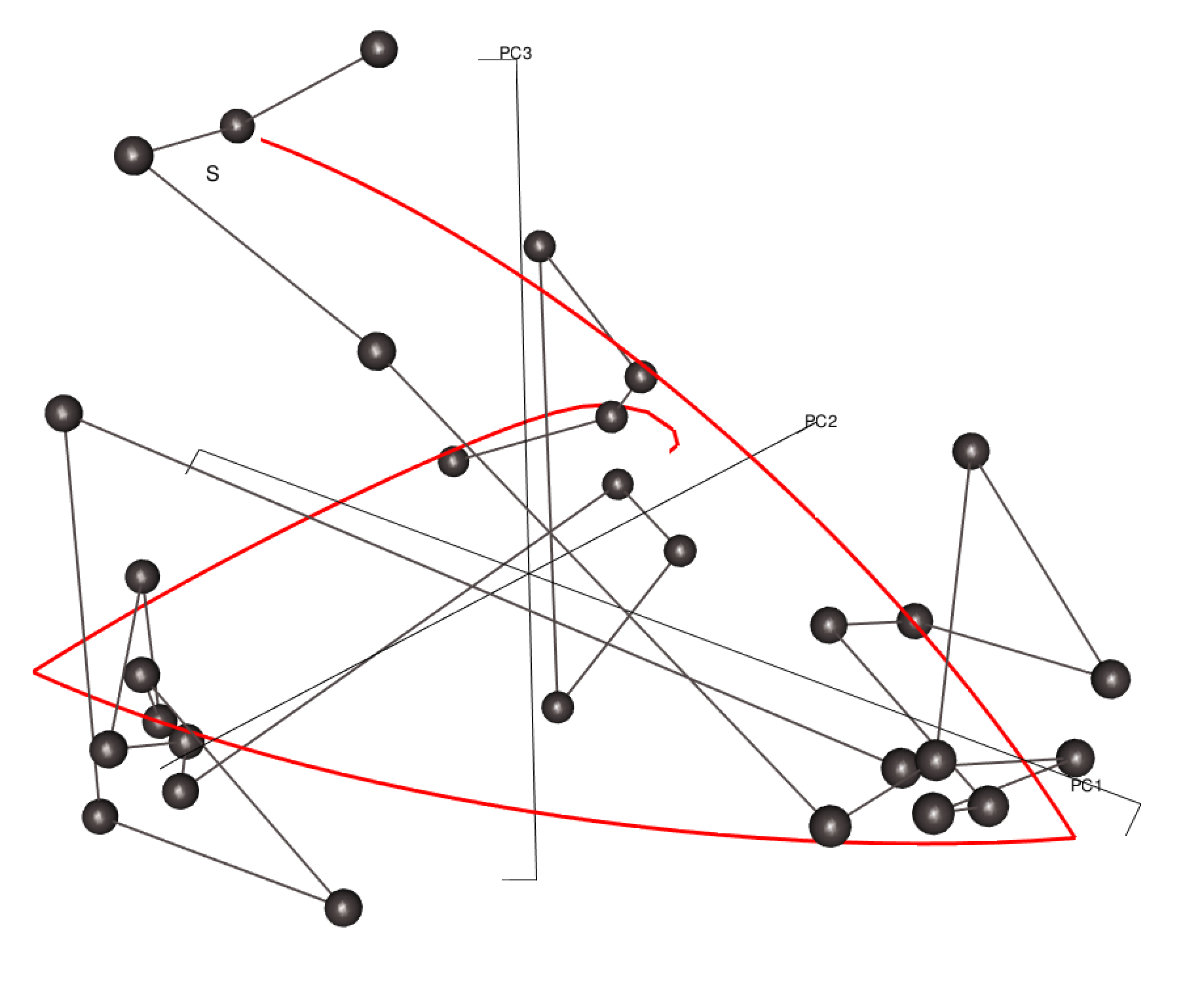}} 
   \end{center}
   \caption{Principal component scores demonstrating (a) the cubic shape smoothing spline fit, (b) single geodesic and (c),(d) linear smoothing spline 
   fits with 3,4 knots. 
   Dots are at PC scores of data points and the solid red thick lines are at PC scores of the fitted values.
   (`S' is located near to the starting point.) }
   \label{fig:peptide.pcscore}
\end{figure}

\subsection{Geodesic and linear spline fitting}
We also compare some alternative models including a single geodesic model and linear splines in the base tangent space with $3,4,5,6,7$ equally spaced knots. We assume 
an independent Gaussian model for errors in the base tangent space with variance $\sigma_i^2$ in the $i$th tangent space dimensions, $i=1,\ldots,M = 87-7 = 80$.  The total number of parameters $p$ in each model 
is given in Table \ref{TABIC}, although for the cubic spline we use the sum of the effective degrees of freedom plus the number of variance parameters. For the information criteria we use
$$ IC = n \sum_i \log \hat \sigma_i^2 +  K_{IC} p $$
where $K_{IC}=2$ for the Akaike Information Criterion (AIC) and $K_{IC} =\log n$ for the Bayesian Information Criterion (BIC), and preferred models have smaller ICs. Here the sample size $n$ is not large compared to $p$ and so we use these criteria 
as informal guides rather than for formal model choice.    

\begin{table}[htbp]
\begin{center}
\begin{tabular}{c|cccc}
Model &  $\sum_i \log(\hat\sigma_i^2)$ & $p$ & AIC & BIC\\
\hline
Geodesic & -259.5 & 240 & -7305 & -6969\\
Linear spline 3 knots & -270.5 & 320 & -7476 & -7028\\
 Linear spline 4 knots & -297.7 & 400 & -8130 & {\bf -7570} \\
 Linear spline 5 knots & -297.9 & 480 & -7977 & -7304\\
 Linear spline 6 knots & -311.8 & 560 & -8235 & -7450 \\
 Linear spline 7 knots & -321.3 &  640 & -8360 & -7463\\
 Cubic spline 31 knots &  -349.4 & 871.7* & {\bf -8739} & -7518\\
 \hline
 \end{tabular}
 \end{center}
  \caption{Model fitting and comparison using AIC and BIC. *Using effective degrees of freedom plus number of variance parameters.}
  \label{TABIC}
 \end{table}
 
 From Table \ref{TABIC} see that the linear spline with 4 knots and the cubic spline are favoured using BIC and AIC respectively, and with BIC 
 there is little to choose between the linear 4 knots model and the cubic spline model.  
 Some of the alternative model fits can also be seen in Figure \ref{fig:peptide.pcscore}(b)-(d) using the PC scores. 
 The PC score fitted paths in (c),(d) are curved as expected as the mean shape of these peptides does not lie in either of these fitted geodesic segments. 
As shown in (b),(c), the geodesic and linear spline 3 knots models do not explain 
the data points well. 
On the other hand, the (a) cubic spline model and (d) linear spline 4 knots model are more reasonable. Both (a) and (d) provide a fit to each of four 
distinct shapes in the data with transition paths inbetween. 
Four preferred states have been observed in this dataset in earlier studies \citep{Drydenetal19}, and further evidence from this analysis that there are four distinct states is valuable confirmation. 

If we had used the Procrustes tangent space at the outset for spline fitting there would have been distortions in the distances between neighbouring configurations, particularly near the end of the sequence, leading to inaccurate predictions. 
In particular, the Procrustes tangent space distance between neighbouring pairs of configurations is around $20-75\%$ larger than the Riemannian distance at times $24,25,26,28$, whereas it is within $10\%$ for the rest of the sequence. 
In the most extreme case between $t=26$ to $t=27$ the Procrustes tangent distance is $0.648$ whereas the Riemannian distance is $0.374$, and we can see extra curvature in the fitted spline near the end of the sequence in the PC scores plot in Figure \ref{fig:peptide.pcscore}(a).

\subsection{Interpolation between states}
An attractive feature of the spline fit is that smooth paths between different states can be explored, to investigate how a peptide transitions in shape from one state to another. For the smooth prediction we have used the data at all the integer times 
but have predicted at $G=2$ equally spaced time points between integer times. 
In Figure \ref{fig:pred} we display the predicted 
shape change in the transition to a state in a later part of the simulation using the cubic spline, at times $t=22$ to $t=25$ at equally spaced intervals. 
We can see that the smooth path predicts the shape change between data points  well, and that the straightening of the peptide is seen in the smoothed predicted path in shape space.

\begin{figure}[htbp]
\addtolength{\subfigcapskip}{-0.2in}
   \centering \subfigure[$t = 22$]
   {\includegraphics[scale = 0.3]{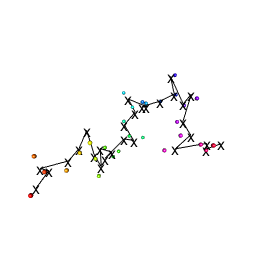}}
   \centering \subfigure[$t = 22.33$]
   {\includegraphics[scale = 0.3]{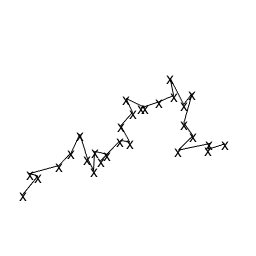}}
     \centering \subfigure[$t = 22.67$]
   {\includegraphics[scale = 0.3]{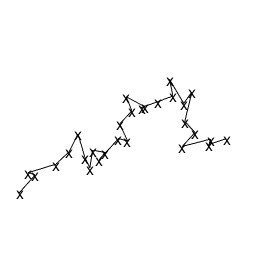}}
   \centering \subfigure[$t = 23$]
   {\includegraphics[scale = 0.3]{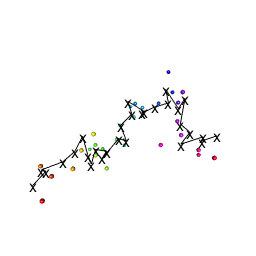}}
   \centering \subfigure[$t = 23.33$]
   {\includegraphics[scale = 0.3]{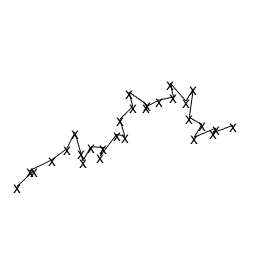}}
     \centering \subfigure[$t = 23.67$]
   {\includegraphics[scale = 0.3]{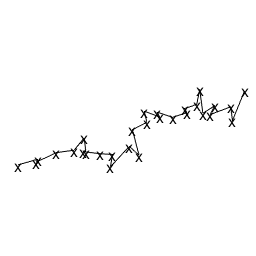}}
     \centering \subfigure[$t = 24$]
   {\includegraphics[scale = 0.3]{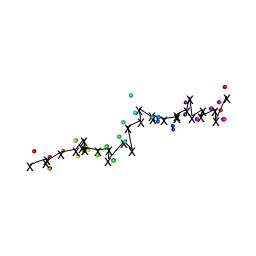}}
      \centering \subfigure[$t = 24.33$]
   {\includegraphics[scale = 0.3]{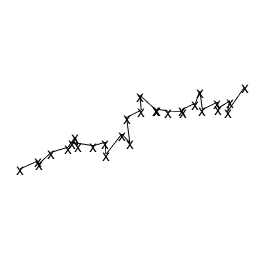}}
      \centering \subfigure[$t = 24.67$]
   {\includegraphics[scale = 0.3]{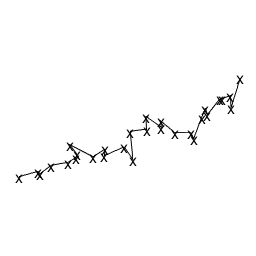}}
       \centering \subfigure[$t = 25$]
   {\includegraphics[scale = 0.3]{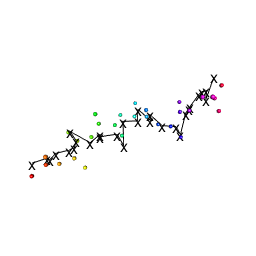}}
   \caption{Configurations of a moving short peptide and predictions for part of the sequence using the cubic spline.
   The rainbow coloured dots indicate 29 landmarks and the connected points `$\times$' are the fitted configurations. 
   Every integer time point has been used for the spline fitting, and the predictions are shown inbetween (at non-integer time points)}
   \label{fig:pred}
\end{figure}

\section{Discussion}\label{Sec5}
An important problem in general in the analysis of molecular dynamics data is the prediction of molecule shape at different parts of the 3D shape space which have not been visited by molecular dynamics simulations \citep[e.g.][]{Laughtonetal09}. Hence, interpolation between visited states (as in our application) is of wide practical interest.

We have chosen to remove size for the peptides and concentrate on shape information alone. The centroid size of the peptide is given in 
Figure \ref{fig:size} and we see that the final part of the sequence is a little larger according to this 
measure. If scale is important to retain than a reasonable approach could be to consider a product of a univariate cubic/linear spline on the scalar centroid size with our spline on the shape space. 
Alternatively the results of Section \ref{sssplines} for size-and-shape space could be used. 

\begin{figure}[htbp]
\begin{center}
      {\includegraphics[width=5cm]{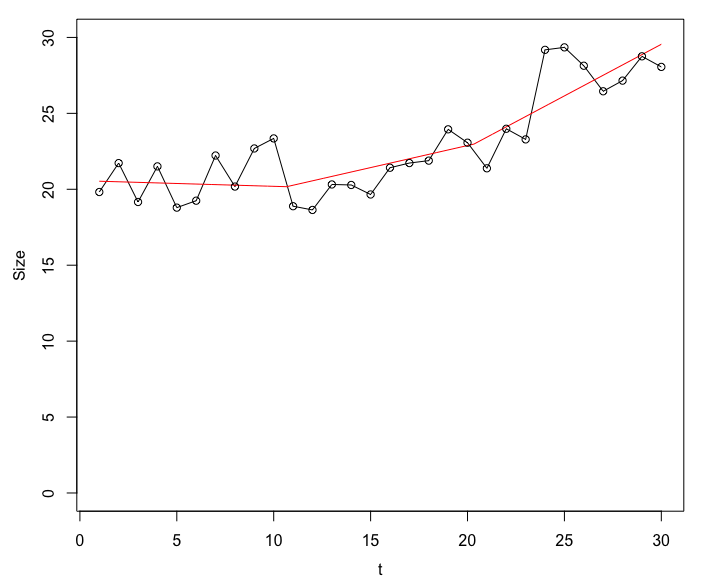}}   
   \end{center}
   \caption{Centroid size for the peptide data (black), and a fitted linear spline with 4 knots (in red). }
   \label{fig:size}
\end{figure}

A further application of shape smoothing spline fitting for some widely varying simulated shape data is given in the Supplementary Material.

Note finally that there are many other applications of smoothing splines on Riemannian manifolds \cite[e.g., see][]{SDKLS}. The main advantage of the unrolling, unwrapping and wrapping method over a simple tangent space method is that 
much larger variations in data can be handled, such as are present in the peptide application. 
{Since it is possible to bypass the need for the knowledge of the curvature tensor in the unrolling technique, 
our implementation of spline fitting on manifolds does not involve the curvature tensor explicitly, which is often difficult to calculate. 
{Hence,} the range of potential applications is very broad indeed.

\section*{Acknowledgement}

This work was supported by the Engineering and Physical Sciences Research Council [grant numbers EP/K022547/1, EP/T003928/1], and
Royal 
Society Wolfson Research Merit Award [grant number WM110140]. We are grateful to the Editor, Associate Editor and two referees for their very helpful comments. 
In particular we are grateful to a referee for suggesting the inclusion of linear splines as part of the work. 

\bibliographystyle{apalike}
\bibliography{model3d}

\newpage 
\subsection*{Appendix} 
\begin{algorithm}
\begin{enumerate}
Consider configurations observed at times $t_0,\ldots,t_n$, with pre-shapes denoted by $X_0,\ldots,X_n$. 
Let $\delta$ and $\varepsilon$ be two given small positive numbers.

\item Rotate the successive pre-shapes, $X_{i+1}$, such that the resulting $X_{i+1}$ is the Procrustes fit of the original one onto $X_i$. That is, rotate the successive pre-shapes $X_{i+1}$ to satisfy the conditions that
$X_iX_{i + 1}^\top$ is symmetric and that all its eigenvalues are non-negative except possibly for the smallest one whose sign is the same as the sign of $\hbox{det}\left(X_iX_{i + 1}^\top\right)$. Denote the set of the resulting pre-shapes by $\D = \{ X_i : 0 \leqslant i \leqslant n \}$ and
let $\gamma_{_\D}$ be the initial piecewise horizontal geodesic path of $\D$ such that
$\gamma_{_\D} (t_i) = X_i$.

\item Construct $G$ grid points between successive two times such that
$$
t_0 = t_{00} < t_{01} < t_{02} < \ldots < t_{0G} <
t_1 = t_{10} < \ldots < t_{ij} < \ldots <
t_n = t_{n0}
$$
which gives $t_{ij}, i = 0, 1, \ldots, n, j = 0, 1, \ldots, G$ for $i \leqslant n - 1$ and $j = 0$ for $i = n$.
where the difference between successive $t_{ij}$ is less than or equal to $\delta$.

\item Set the base path $\gamma_{_{\D_1}}$ to be the piecewise horizontal geodesic passing through $\D_1 = \{ \gamma_{_\D} (t_{ij}) : \forall i, j \}$.

\item Using the {unrolling and unwrapping} procedures  described in Section \ref{Sec2}, with the parallel transport $P$ there replaced by $\Psi$ defined by \eqref{eqn5a}, using the expression \eqref{eqn0c} for the inverse exponential, and using the approximation procedure for $\Psi$ described at the end of 
Section \ref{ParallelTransport}, unwrap the data $\D$ into $\tau_{\gamma_{_{\D_1}} (t_0)} (\S_m^k)$
with respect to the base path $\gamma_{_{\D_1}}$
which gives 
\[\D^\dagger := \{X_0^\dagger, X_1^\dagger, \ldots, X_n^\dagger \} \subset \H_{\gamma_{_{\D_1}} (t_0)} (\S_m^k).\]

\item Fit the cubic smoothing spline (\ref{cubicspline}) to $\D^\dagger$
and find fitted values at the times of the grid points
given in Step 2, giving 
\[\D_2^\dagger := \{ Z_{ij}^\dagger = \widehat{f^\dagger}(t_{ij}, \hat\lambda) : i, j \} \subset \H_{\gamma_{_{\D_1}} (t_0)} (\S_m^k).\]

\item Using the {wrapping} procedure {described in Section \ref{Sec2} with the parallel transport $P$ there replaced by $\Psi$ defined by \eqref{eqn5a}, using \eqref{eqn0d} for the exponential map on the sphere $\mathcal S_m^k$, and using the approximation procedure for $\Psi$ described at the end of 
Section \ref{ParallelTransport},} wrap $\D_2^\dagger$ back into the pre-shape sphere
with respect to the base path $\gamma_{_{\D_1}}$, giving $\D_2 := \{ Z_{ij} : i, j \} \subset \S_m^k$.

\item Successively rotate $Z_{ij}$ in $\D_2$ such that the resulting $Z_{i\,j+1}$ is the Procrustes fit of the original one onto $Z_{ij}$, 
and obtain the piecewise horizontal geodesic $\gamma_{_{\D_2}}$ passing through the resulting $Z_{ij}$.
Then, $\gamma_{_{\D_2}}$ becomes the base path in the next iteration.

\item If $\max \{ d([\gamma_{_{\D_2}} (t_{ij})], [\gamma_{_{\D_1}} (t_{ij})]) : i, j \} \geqslant \varepsilon$,
replace $\D_1$ by $\D_2$ (where $d$ is the shape distance); 
successively rotate $X_i$ in $\D$ such that the resulting $X_i$ is the Procrustes fit of the original one onto $Z_{i0}$; update $\D$ by the resulting $X_i$; and then go to Step 4.
Otherwise, stop the iterations.

\end{enumerate}
\end{algorithm}

\newpage

{\bf \Large Supplementary Material to `Smoothing splines on Riemannian manifolds, with applications to 3D shape space' by Kim, Dryden, Le and Severn} 

\section*{Simulated data}
We apply the shape smoothing spline methodology to some widely varying simulated datasets of configurations in $\mathbb{R}^3$ with $k = 8$ landmarks. Four initial vertex configurations are chosen such that the Riemannian shape distances between successive configurations are 0.47, 0.75 and 0.54 respectively. Note that the Riemannian shape distance has full range $[0, \pi / 2]$, so that the distances 
between the shapes of these consecutive configurations are quite large and the full path covers a large distance in shape space. 
Then four equally spaced shapes on each of the three geodesic segments connecting the successive shapes of the vertex configurations are generated; the configurations of these 16 shapes are shown in Figure  \ref{fig:conf.true}, where the 1st, 6th, 11th and 16th are the four initial configurations. 
Then, we add Gaussian noise with standard deviation $\sigma = 0.05$ independently to each landmark of the original configurations to generate perturbed data which is shown in Figure \ref{fig:conf.simul}.
The configurations in these examples vary considerably in shape and so this provides a challenging setting for shape spline fitting. 

In more detail let the original configurations be  $\mu_i, i=0,1,\ldots,n$ with shapes that lie on a continuous piecewise geodesic path of three geodesics 
in shape space, and the $\mu_i$ are $m \times k$ matrices with $k=8, m=3, n=15$. 
These original configurations are displayed in Figure \ref{fig:conf.true}. The model for the perturbed data is 
$$ Y_i = \mu_i + E_i \; , \; i=0,1,\ldots,n, $$
where the $(r,s)$th element of  $E_i$ is independently $(E_i)_{rs} \sim N(0,\sigma^2)$, $r=1,\ldots,m$, $s=1,\ldots,k$ 
and $\sigma=0.05$.  These perturbed configurations are displayed in Figure \ref{fig:conf.simul}.

\begin{figure}[p]
   \centering
   {\includegraphics[scale = 0.17]{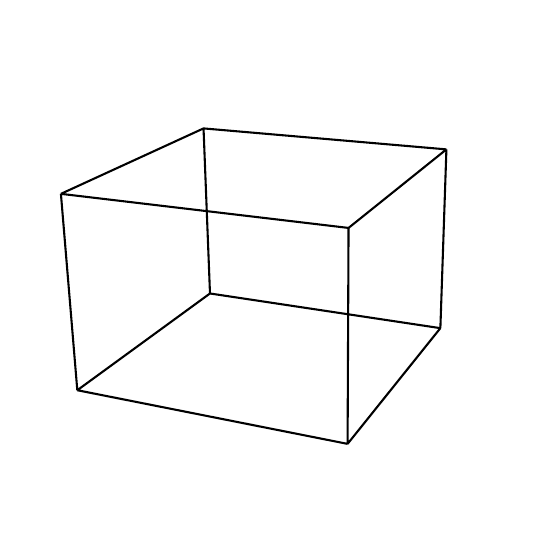}}
   \centering
   {\includegraphics[scale = 0.17]{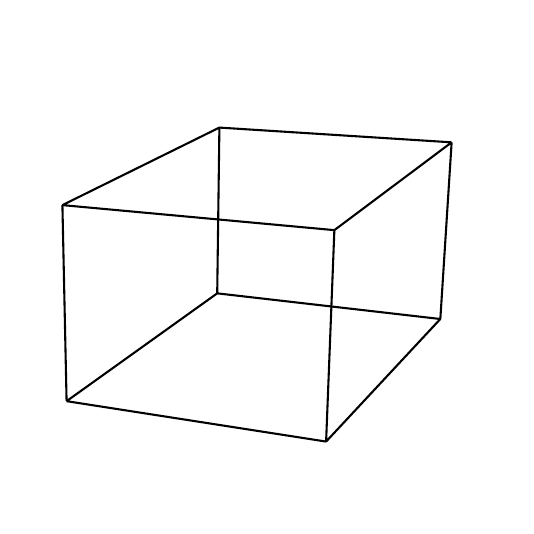}}
   \centering
   {\includegraphics[scale = 0.17]{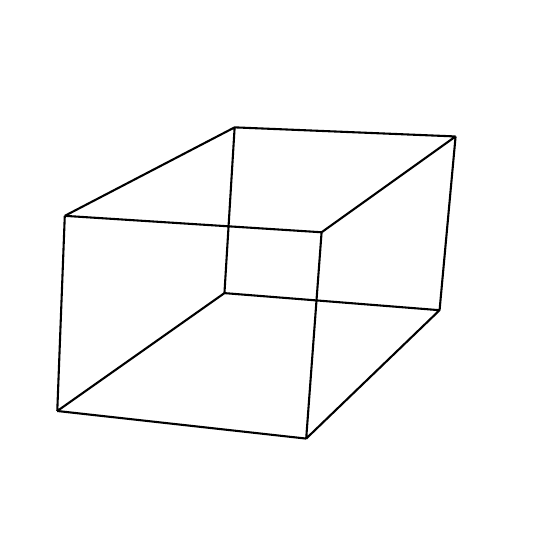}}
   \centering
   {\includegraphics[scale = 0.17]{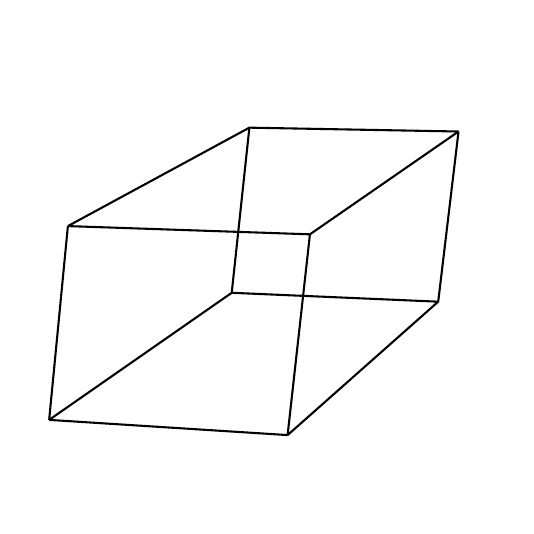}}
   \centering
   {\includegraphics[scale = 0.17]{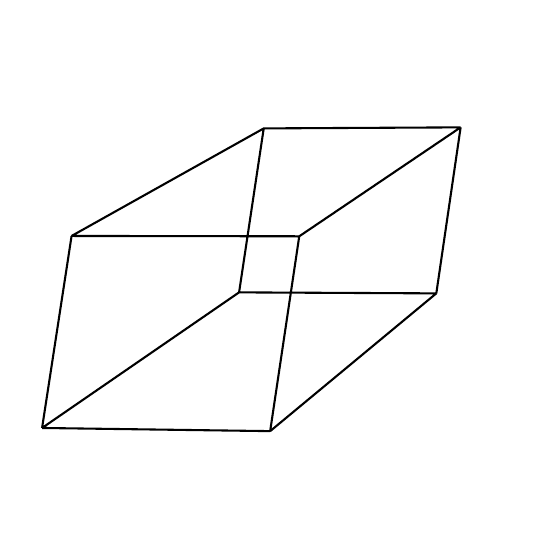}}
   \centering
   {\includegraphics[scale = 0.17]{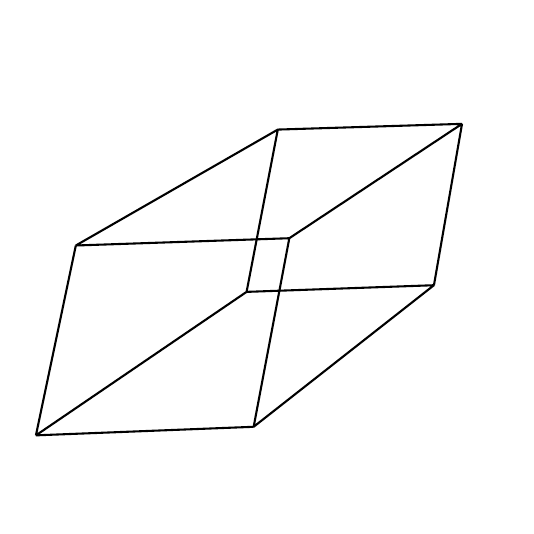}}
   
   \centering
   {$\qquad\qquad$
   \includegraphics[scale = 0.17]{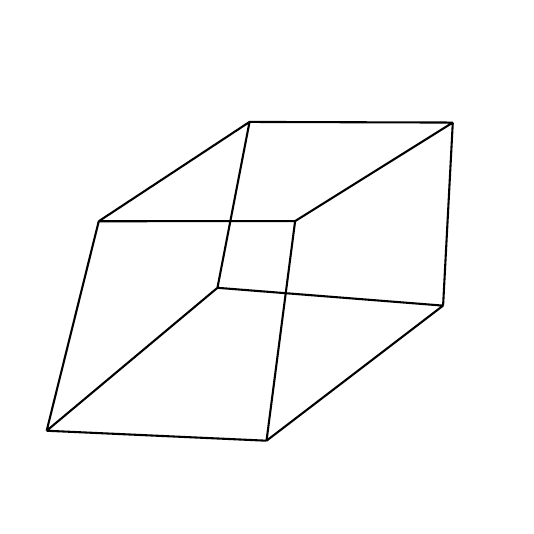}
   \includegraphics[scale = 0.17]{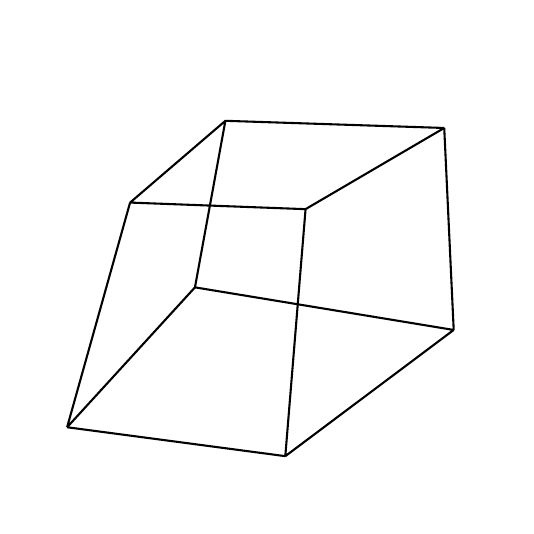}
   \includegraphics[scale = 0.17]{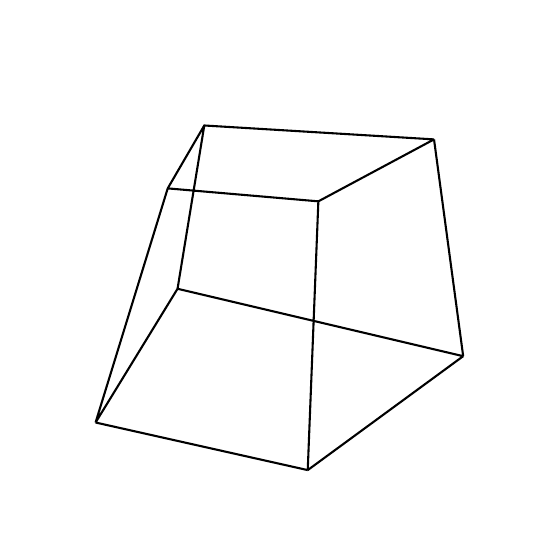}
   \includegraphics[scale = 0.17]{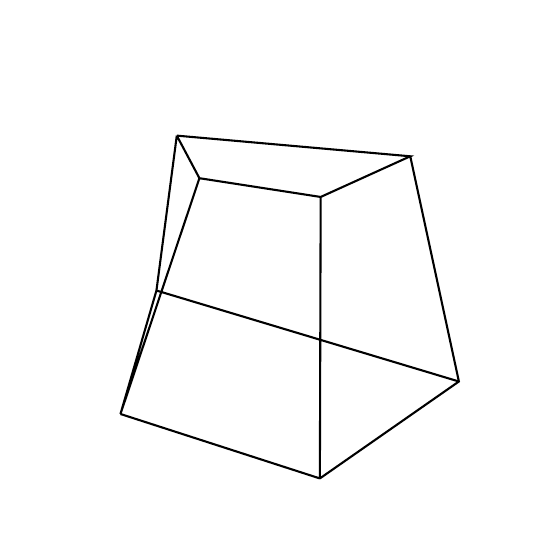}
   \includegraphics[scale = 0.17]{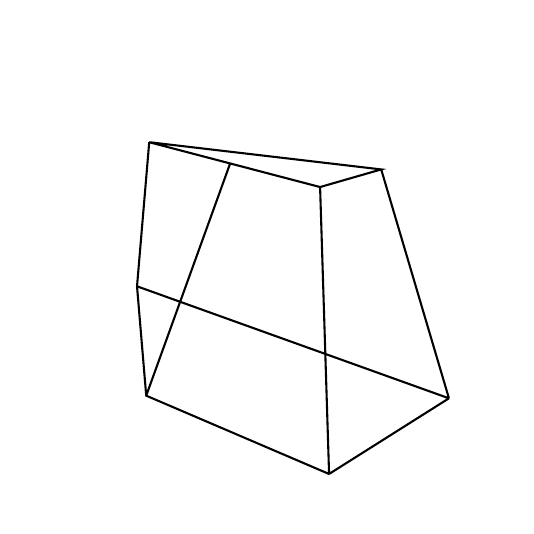}}
   
   \centering
   {$\qquad\qquad$ 
   \includegraphics[scale = 0.17]{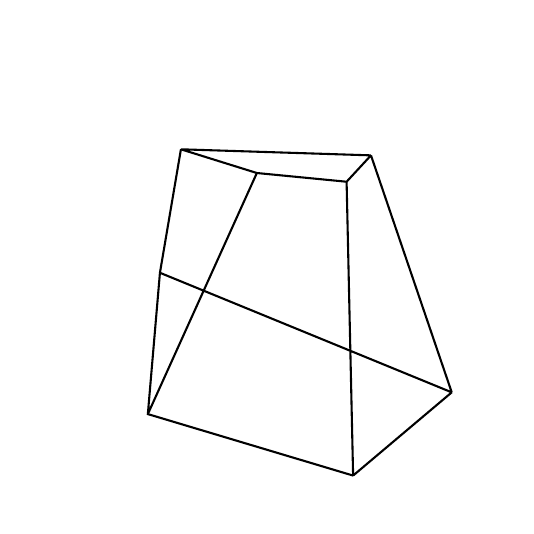}
   \includegraphics[scale = 0.17]{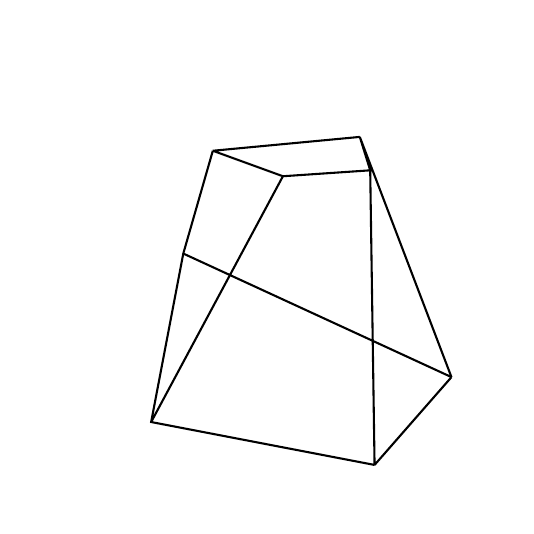}
   \includegraphics[scale = 0.17]{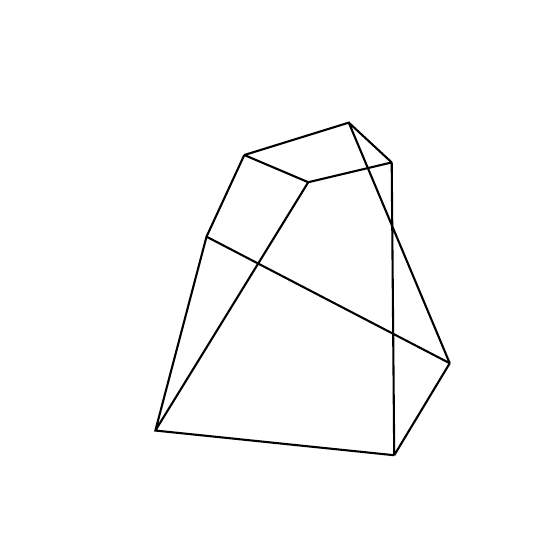}
   \includegraphics[scale = 0.17]{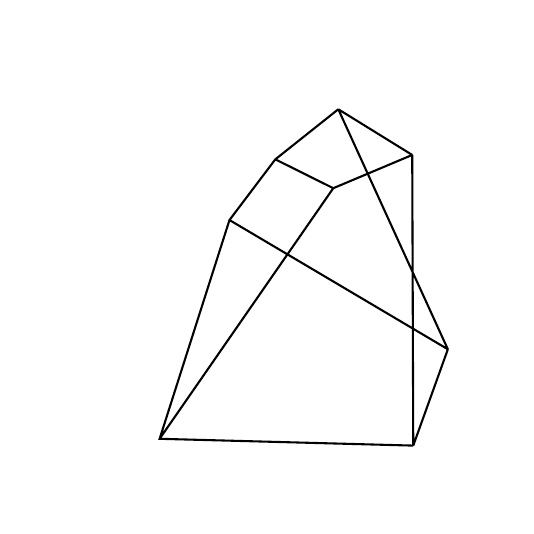}
   \includegraphics[scale = 0.17]{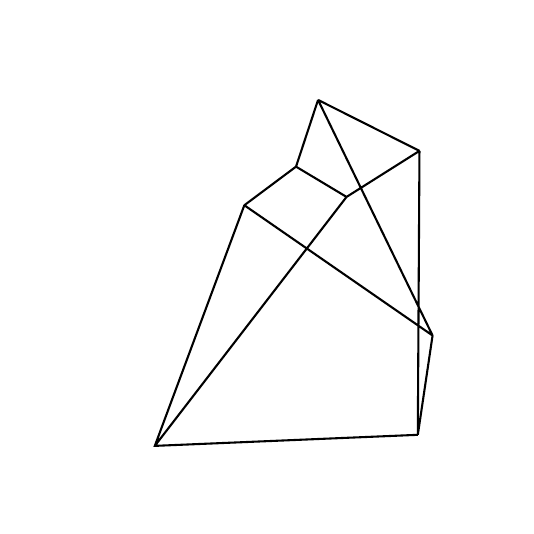}}
   \caption{Original configurations.}
   \label{fig:conf.true}
\end{figure}
\begin{figure}[htbp]
   \centering
   {\includegraphics[scale = 0.17]{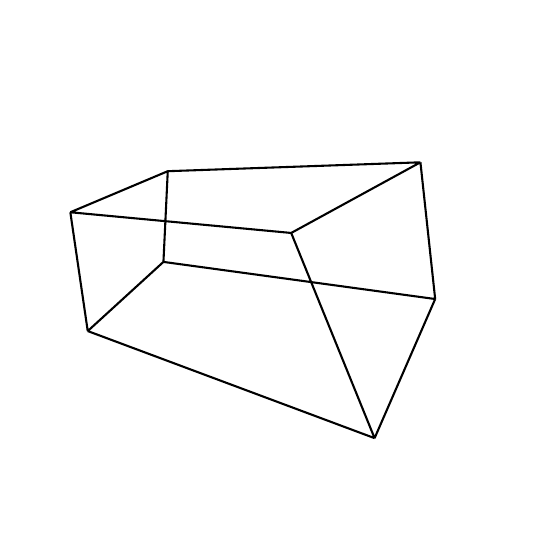}}
   \centering
   {\includegraphics[scale = 0.17]{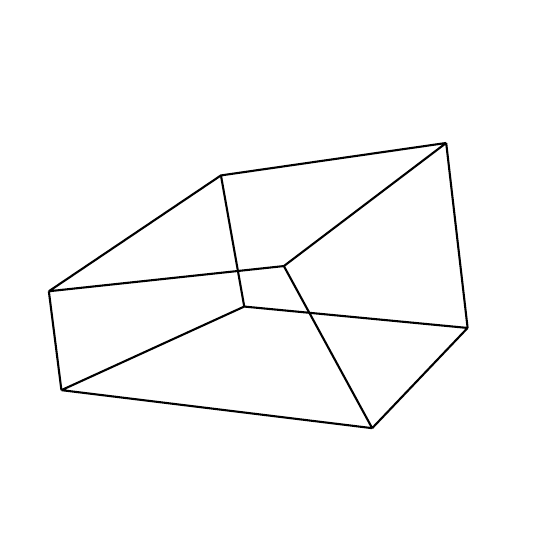}}
   \centering
   {\includegraphics[scale = 0.17]{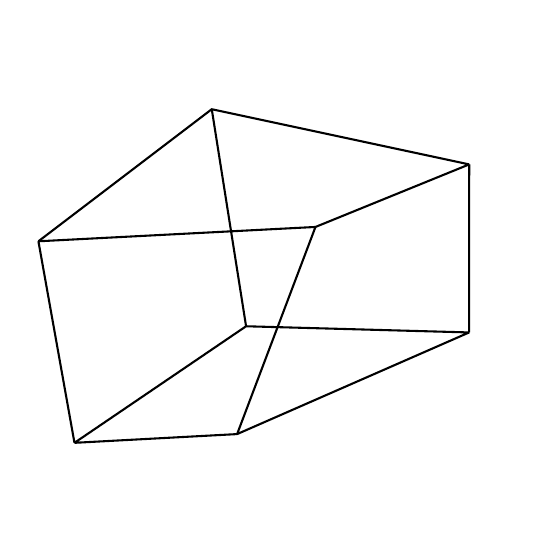}}
   \centering
   {\includegraphics[scale = 0.17]{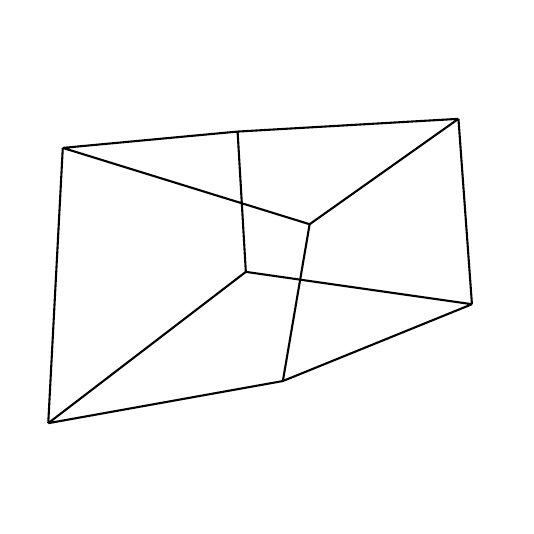}}
   \centering
   {\includegraphics[scale = 0.17]{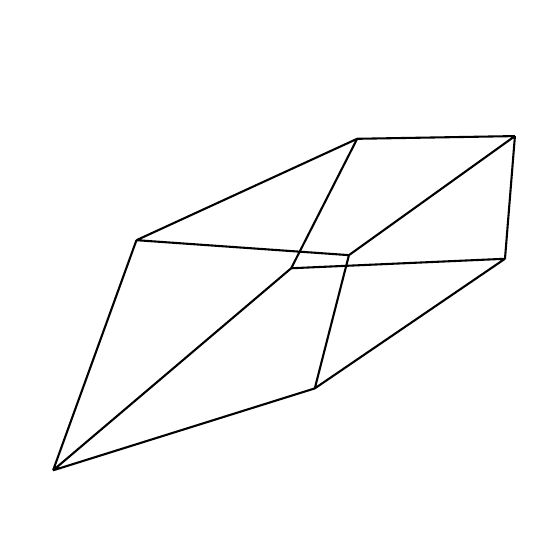}}
   \centering
   {\includegraphics[scale = 0.17]{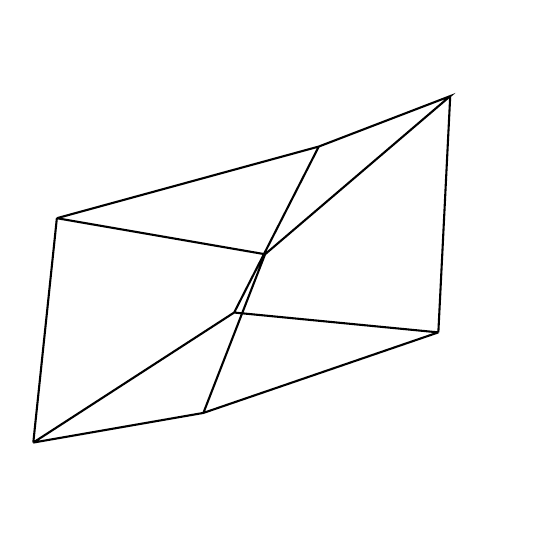}}

   \centering
   {$\qquad\qquad$
   \includegraphics[scale = 0.17]{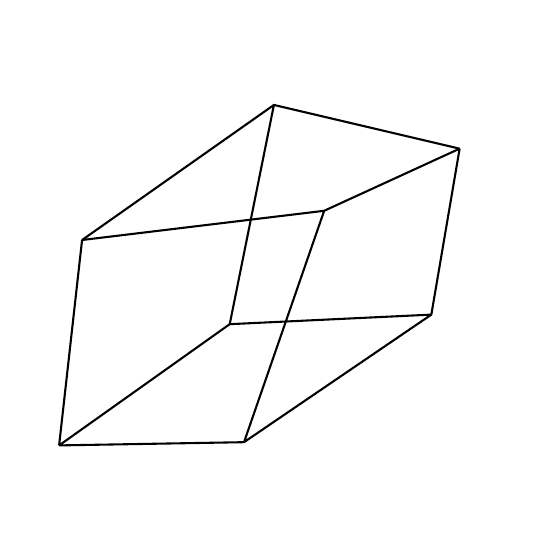}
   \includegraphics[scale = 0.17]{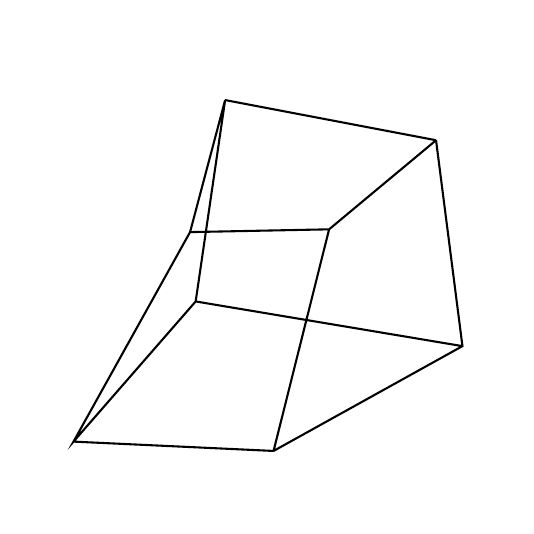}
   \includegraphics[scale = 0.17]{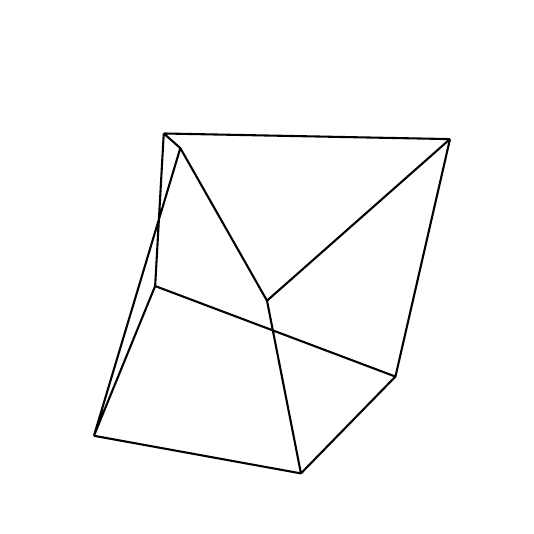}
   \includegraphics[scale = 0.17]{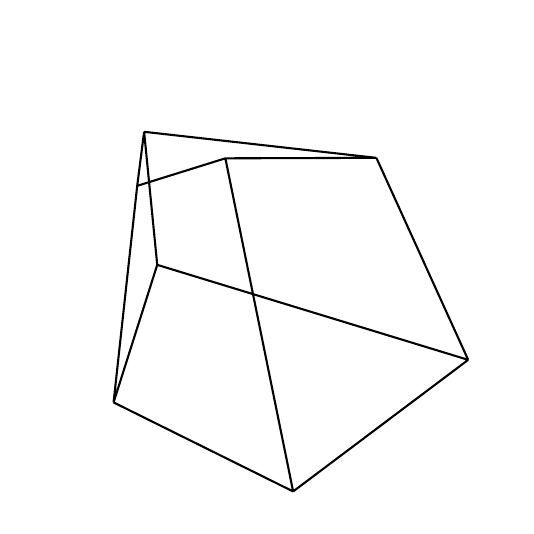}
   \includegraphics[scale = 0.17]{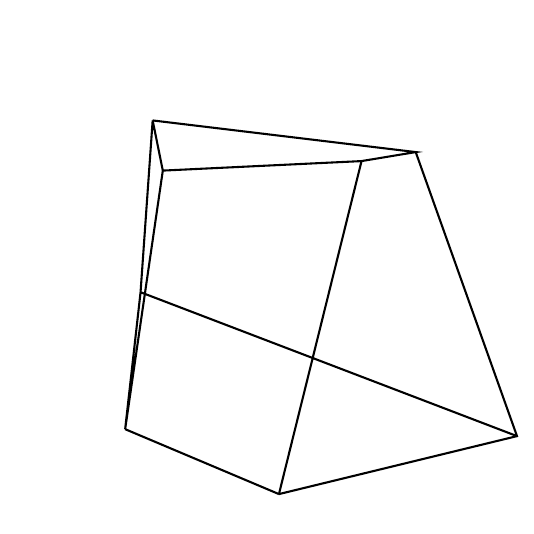}}
   
   \centering
   {$\qquad\qquad$
   \includegraphics[scale = 0.17]{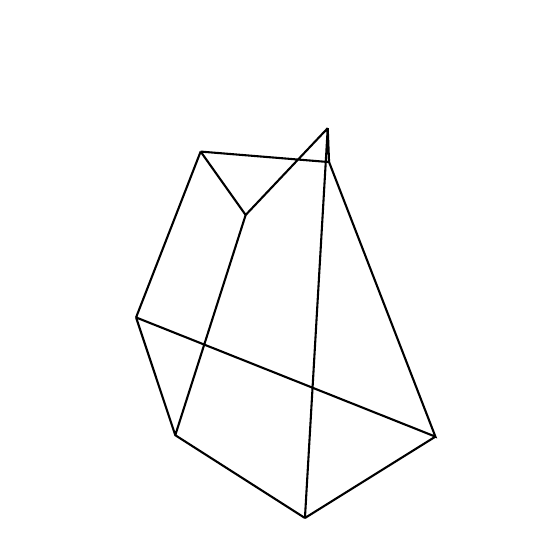}
   \includegraphics[scale = 0.17]{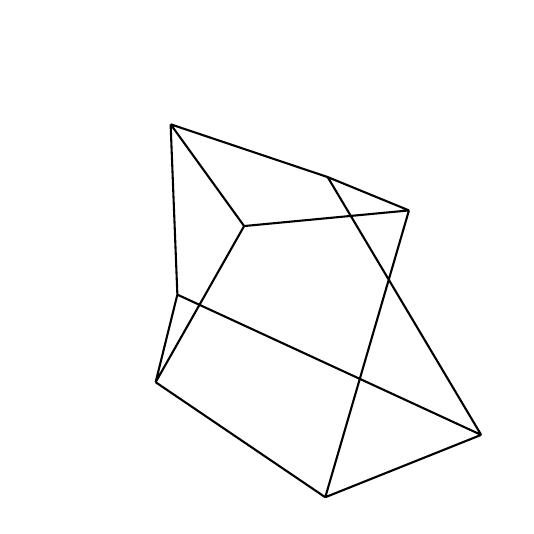}
   \includegraphics[scale = 0.17]{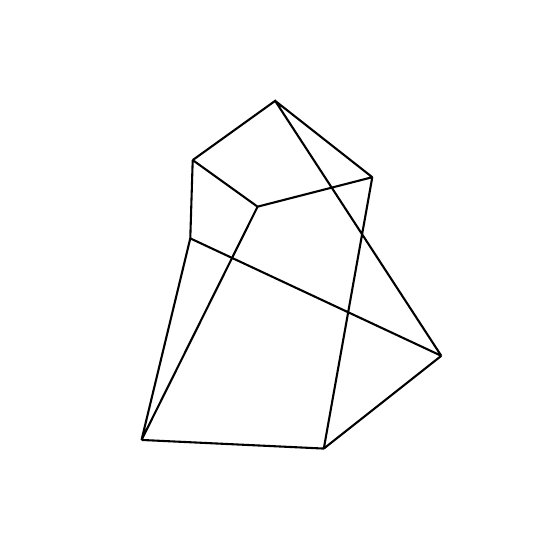}
   \includegraphics[scale = 0.17]{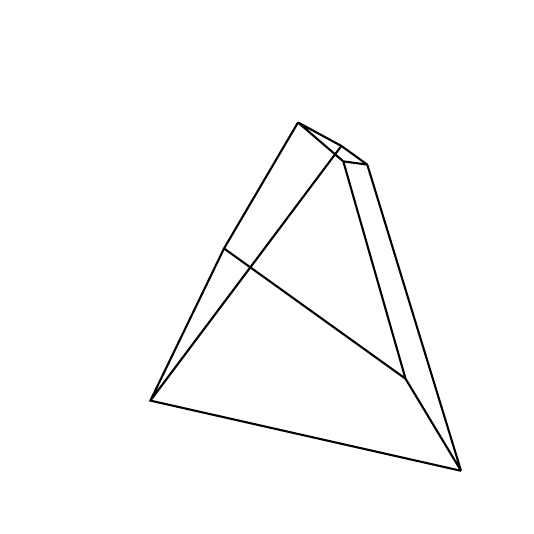}
   \includegraphics[scale = 0.17]{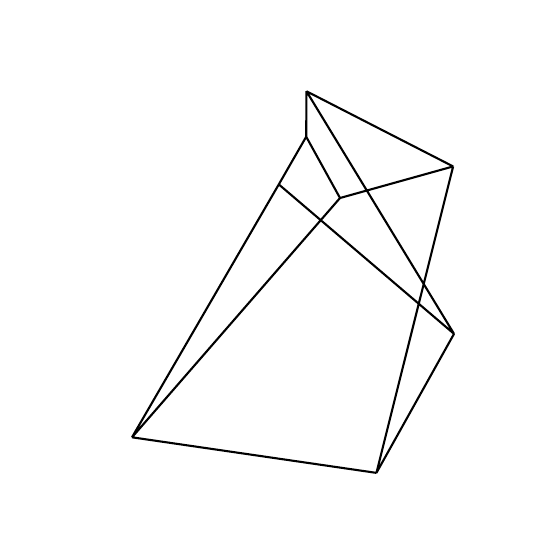}}
   \caption{Perturbed configurations.}
   \label{fig:conf.simul}
\end{figure}
\begin{figure}[htbp]
   \centering
   {\includegraphics[scale = 0.17]{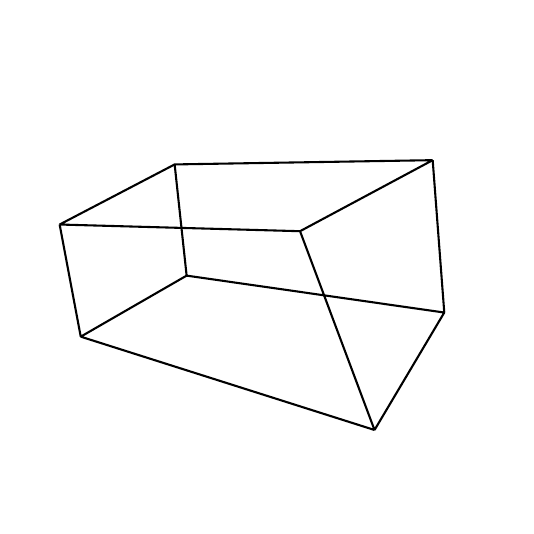}}
   \centering
   {\includegraphics[scale = 0.17]{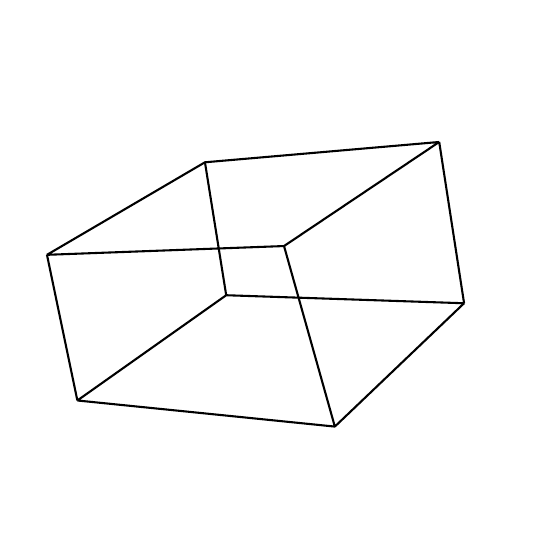}}
   \centering
   {\includegraphics[scale = 0.17]{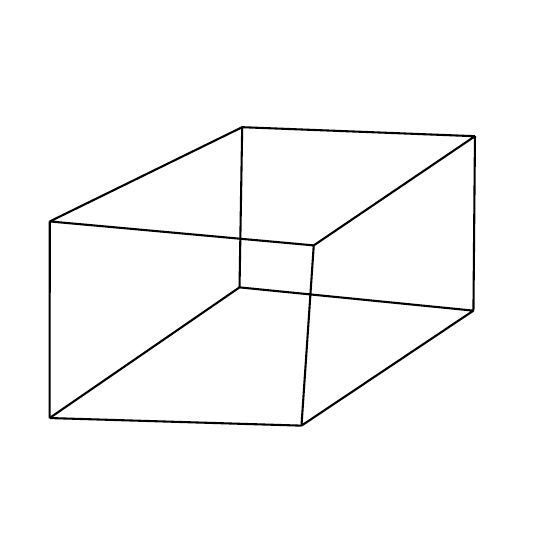}}
   \centering
   {\includegraphics[scale = 0.17]{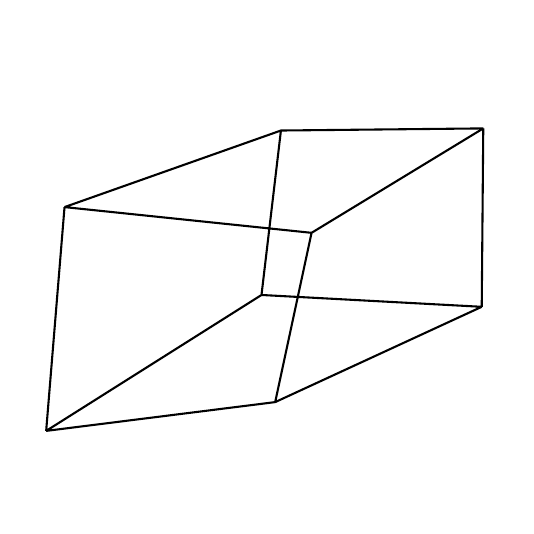}}
   \centering
   {\includegraphics[scale = 0.17]{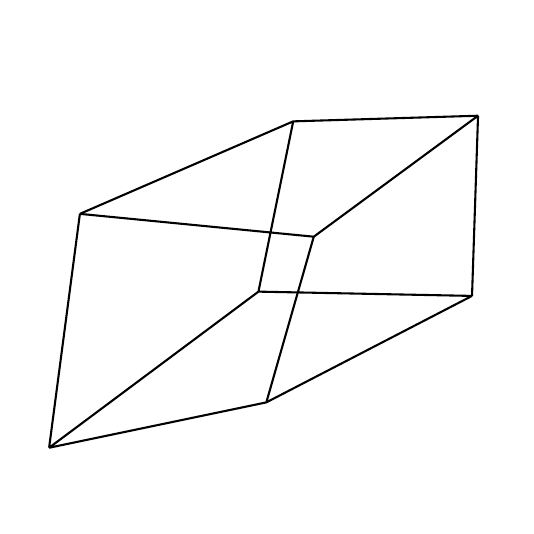}}
   \centering
   {\includegraphics[scale = 0.17]{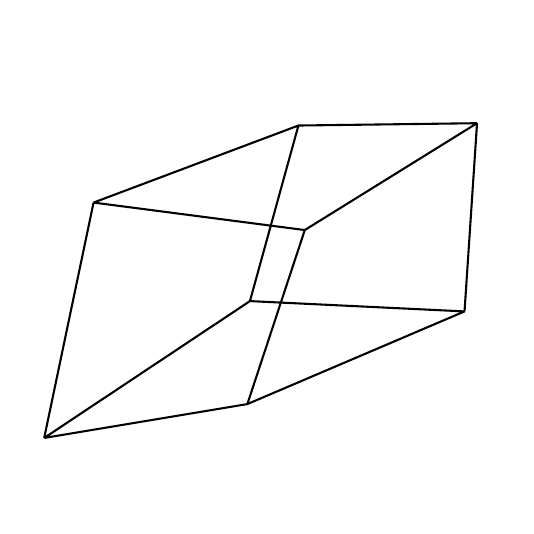}}
   
   \centering
   {$\qquad\qquad$   
   \includegraphics[scale = 0.17]{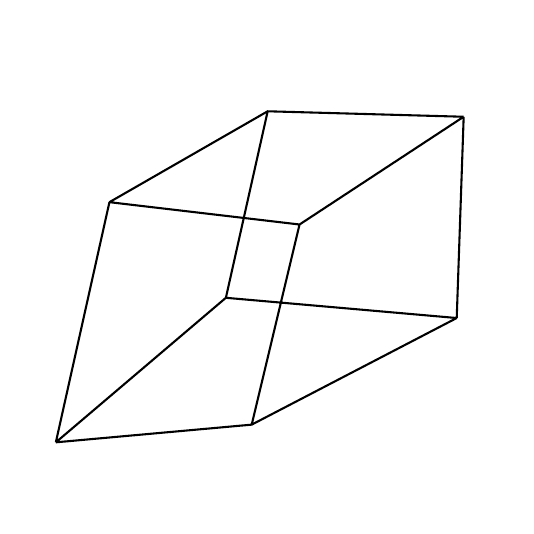}
   \includegraphics[scale = 0.17]{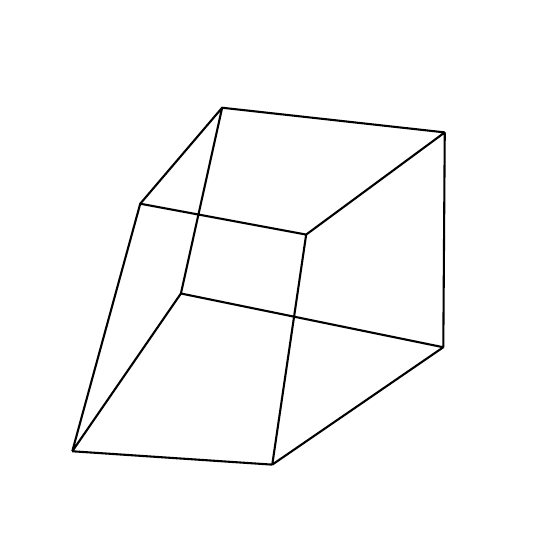}
   \includegraphics[scale = 0.17]{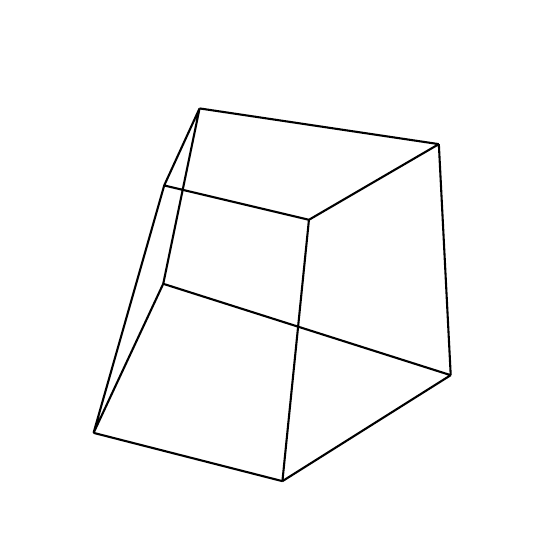}
   \includegraphics[scale = 0.17]{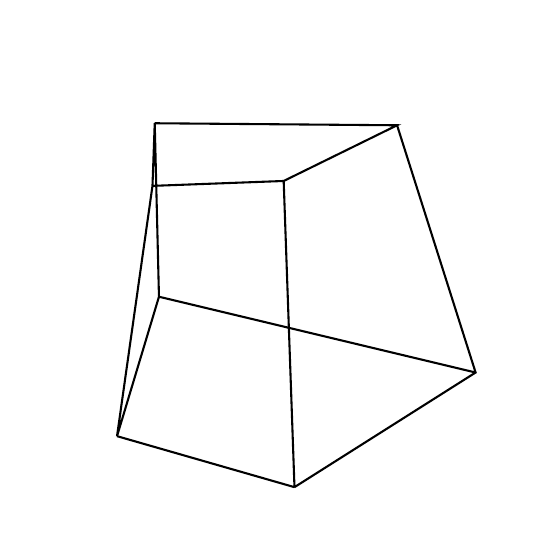}
   \includegraphics[scale = 0.17]{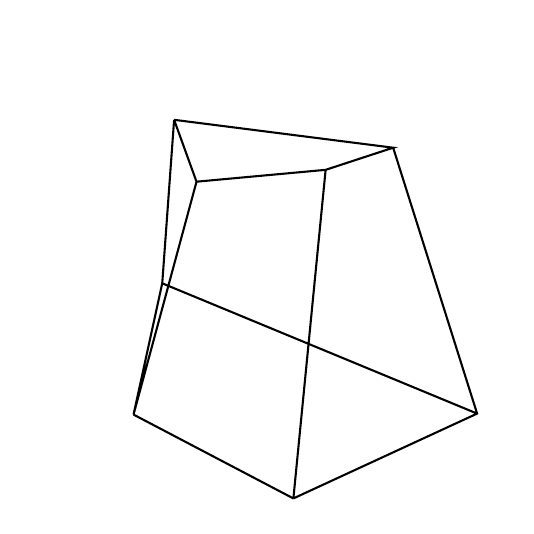}}
   
   \centering
   {$\qquad\qquad$
   \includegraphics[scale = 0.17]{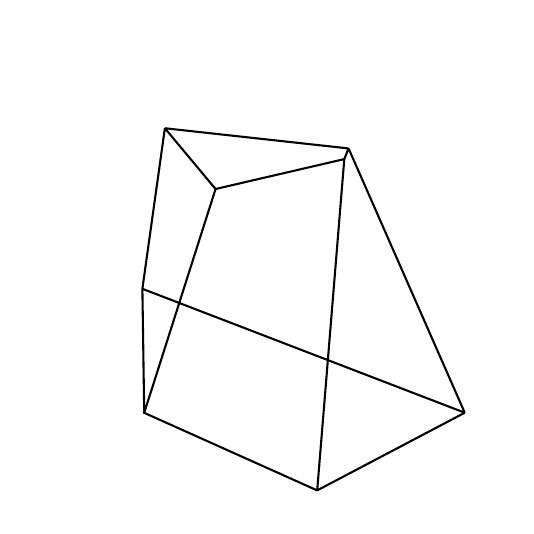}
   \includegraphics[scale = 0.17]{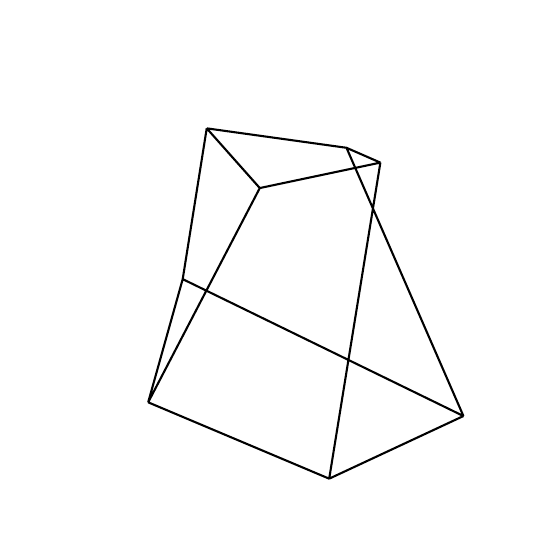}
   \includegraphics[scale = 0.17]{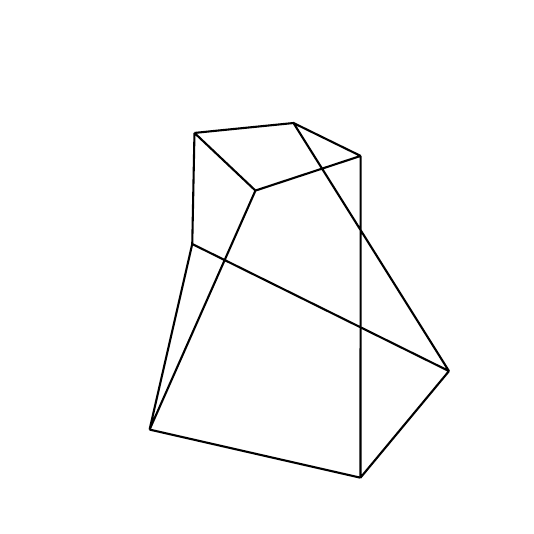}
   \includegraphics[scale = 0.17]{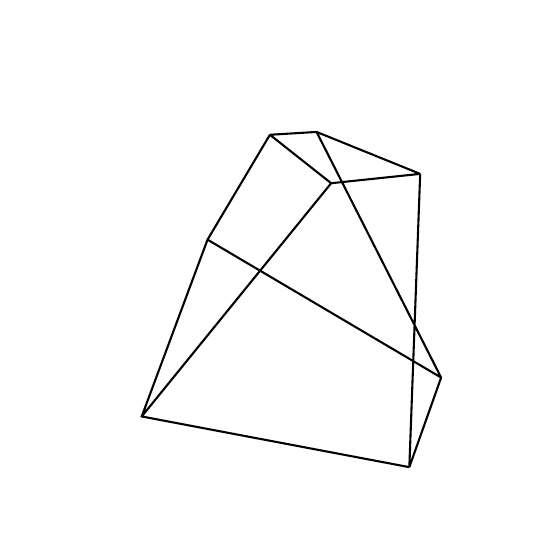}
   \includegraphics[scale = 0.17]{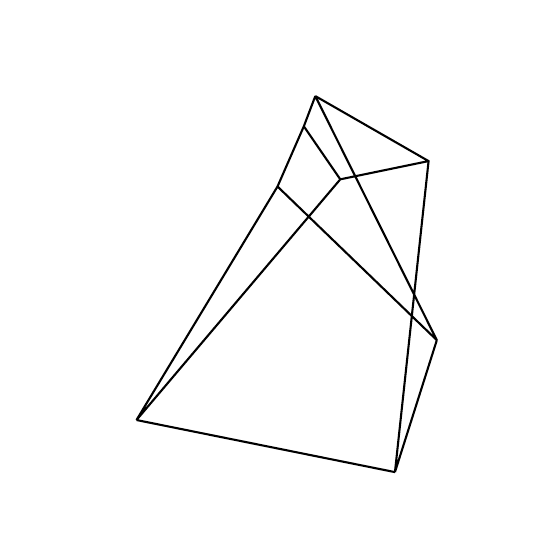}}
   \caption{Configurations of the fitted shapes to the shapes of the perturbed configuration data.}
   \label{fig:conf.simu.fit}
\end{figure}

We obtain the pre-shapes $X_i$ of the configurations $Y_i$ by post-multiplying by transpose of the Helmert sub-matrix $H$ \citep[p49]{Drydmard16} 
and rescaling:  
$$ {X}_i = \frac{ {Y_i} H^T }{  \| {Y_i} H^T \| } , \; i=0,1,\ldots,n. $$
We now apply the iterative shape spline fitting algorithm using unrolling, unwrapping and wrapping to fit the smoothing spline to both 
the original data (the preshapes of $\mu_i$) and the perturbed pre-shape data $X_i$, $i=0,1,\ldots,n$. 
The aim is to predict the smooth underlying path of the original configuration using the perturbed data. 
The candidates for $\lambda$ in the cross-validation were taken from the set $\{10^{-9}, 10^{-7}, 10^{-5}, 10^{-3}, 10^{-1} \}$. 
Consequently the chosen smoothing parameters for the original and perturbed data are
$\lambda=10^{-5}$ and $\lambda=10^{-3}$, respectively.
Figure \ref{fig:conf.simu.fit} shows configurations corresponding to the fitted pre-shapes on the shape smoothing spline to the perturbed data. 
We can see clearly that the fitted configurations form a smoother path compared to the perturbed data, and the fitted configurations 
are closer in shape to the original data than they are to the perturbed data.

We use principal components analysis to provide a low dimensional visualization of the fitted shape shape smoothing splines to the original and perturbed data. 
In each case we unroll the given 16 shapes onto the base tangent space to the first shape along the piecewise geodesic segments through them and unroll the fitted shape smoothing spline onto the tangent space to its starting point. Then, we carry out the principal component analysis on the coordinates of these four unrolled data sets respectively. Figure \ref{fig:pcscore.simul} 
\begin{figure}[htbp]
   \centering \subfigure[Original data]
   {\includegraphics[scale = 0.58]{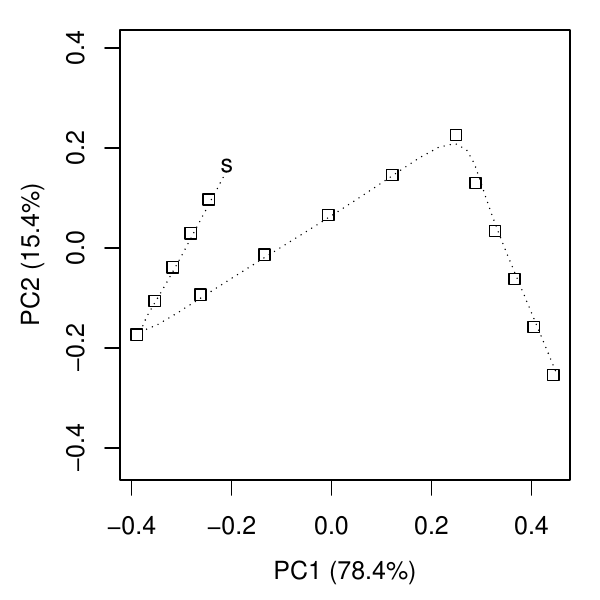}}
   \centering \subfigure[Perturbed Gaussian ($\sigma=0.05$)]
   {\includegraphics[scale = 0.58]{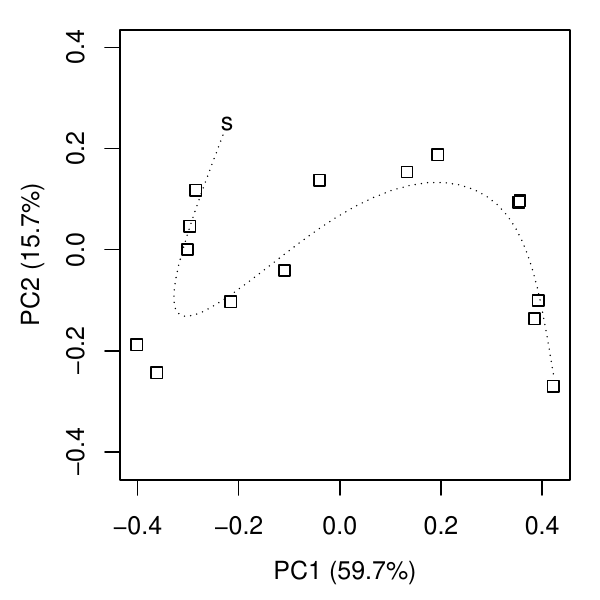}}
   \caption{The first two principal component scores of the unrolled shape data and unrolled fitted shape spline in the tangent spaces at the starting points.}
   \label{fig:pcscore.simul}
\end{figure}
shows the first two principal component scores of the unrolled original data and the corresponding fitted shape spline, in (a),
and the unrolled perturbed data ($\sigma = 0.05$) and the corresponding fitted shape spline, in (b). The given data are displayed as squares
and the dotted lines are for the fitted shape spline with `s' indicating the starting point. For the original non-perturbed data set, the proportion of variance explained by PC1 and PC2 is 78.4\% and 15.4\% respectively and, for the fitted shape spline to the perturbed data, 59.7\% and 15.7\% are explained by PC1 and PC2.  Although the path of the configurations lies close to a plane, it is not exactly in a plane. For 
the original data the percentage of variability in the first two tangent space PCs is 93.8\%, and so the remainder of the variability is in the third PC (6.2\%) which is small and so we ignore this dimension in our analysis.   The percentage variability is included in our analyses throughout to demonstrate that
most of the variability in the data is captured in the first few PCs. 
Noting that the two turning points of these three geodesic segments are at the 6th and 11th configurations, Figure \ref{fig:pcscore.simul}(a) clearly shows the structure of the three geodesic segments which is in accord with the simulation setting. For the perturbed data set, the overall pattern is similar to that of the original data but with local differences due to the noise. We can see that the shape smoothing spline has provided good predictions, with the
structure of the paths being similar between the fitted path through the perturbed data and the original data. 

Further examples are also given in Figure \ref{fig:pcscore.simul2} with different independent Gaussian noise levels (a) $\sigma = 0.02$, (b) $\sigma = 0.1$, and independent Student's $t_3$ noise with standard deviation (c)  
$\sigma = 0.02$ and (d) $\sigma = 0.05$. In the cases (a)-(c) the three segments can be seen in the fitted paths. The final plot (d) contains a large outlier and the structure is harder to see.   Note that with more noise in (b) and (d) the 
first two PCs explain a smaller percentage of the shape variability (57.5\% and 71.2\% respectively for (b) and (d)). 

\begin{figure}[htbp]
   \centering \subfigure[Perturbed Gaussian ($\sigma=0.02$)]
   {\includegraphics[scale = 0.58]{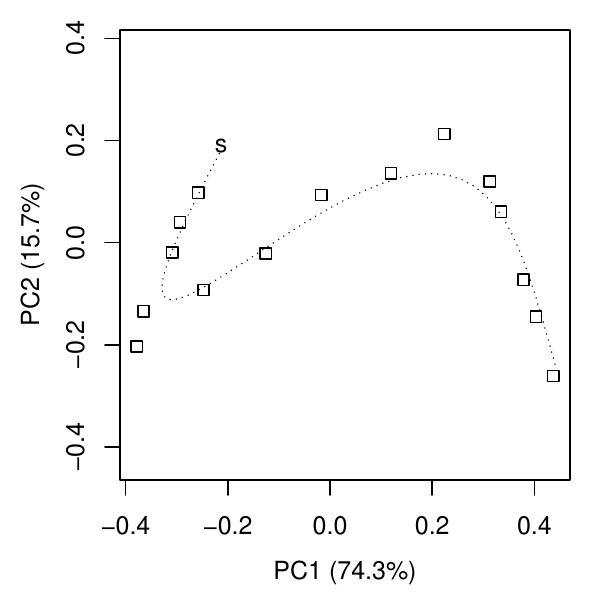}}
   \centering \subfigure[Perturbed Gaussian ($\sigma=0.1$)]
   {\includegraphics[scale = 0.58]{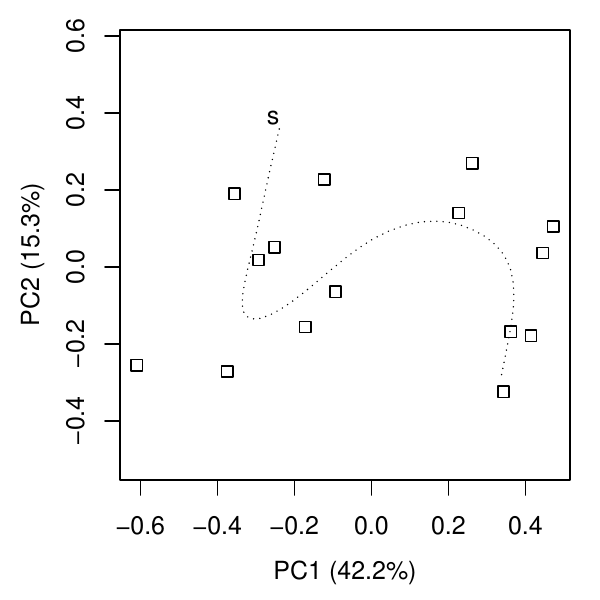}}
   \centering \subfigure[Perturbed $t_3$ ($\sigma=0.02$)]
   {\includegraphics[scale = 0.58]{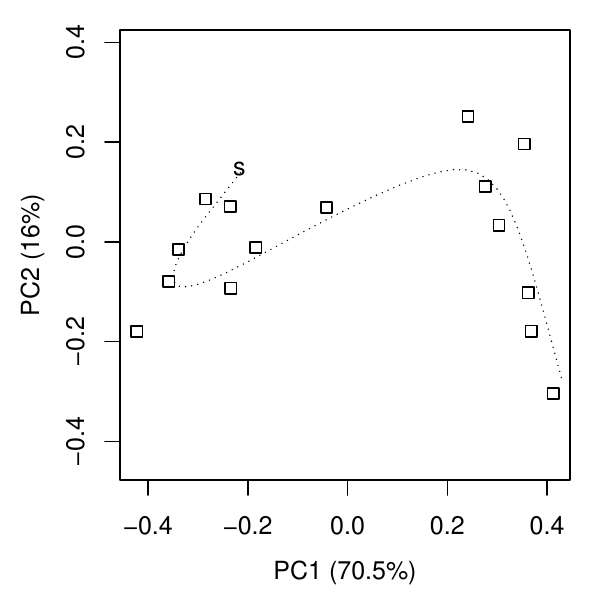}}
    \centering \subfigure[Perturbed $t_3$ ($\sigma=0.05$)]
   {\includegraphics[scale = 0.58]{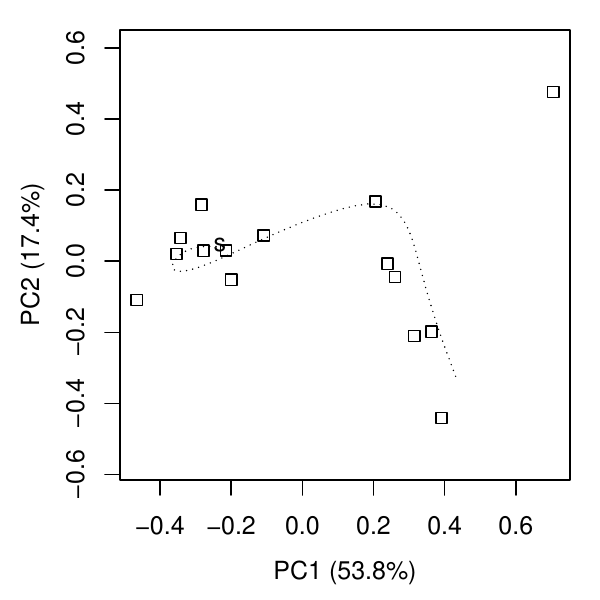}}  
   \caption{The first two principal component scores of the unrolled shape data and unrolled fitted shape spline in the tangent spaces at the starting points for independent (a) Gaussian noise levels (a) $\sigma = 0.02$, (b) $\sigma = 0.1$, and independent Student's $t_3$ noise with standard deviation (c)  
$\sigma = 0.02$ and (d) $\sigma = 0.05$   
   }
   \label{fig:pcscore.simul2}
\end{figure}

These simulations demonstrate that the shape smoothing splines can
be used for fitting smooth paths for widely varying 3D data and give appropriate results when noise is present.

\end{document}